\documentclass[11pt,a4paper,reqno]{article}
\usepackage{mysty}
\usepackage{authblk}

\usepackage{fullpage}
\usepackage{subfig}

\usepackage{amsfonts,amssymb}
\usepackage{amsmath}
\usepackage{graphicx}
\usepackage{cite}
\usepackage{enumerate}
\usepackage{esint}

\theoremstyle{plain}
\newtheorem{thm}{\protect\theoremname}
\theoremstyle{plain}

\theoremstyle{remark}

\usepackage{babel}
\providecommand{\lemmaname}{Lemma}
\providecommand{\theoremname}{Theorem}
\providecommand{\remarkname}{Remark}

\SetLabelAlign{parright}{\parbox[t]{\labelwidth}{\raggedleft{#1}}}
\setlist[description]{style=multiline,topsep=4pt,align=parright}

\makeatletter
\let\reftagform@=\tagform@
\def\tagform@#1{\maketag@@@{(\ignorespaces\textcolor{black}{#1}\unskip\@@italiccorr)}}
\newcommand{\iref}[1]{\textup{\reftagform@{\tcr{\ref{#1}}}}}
\makeatother

\begin{document}
	
\title{Quantitative relations among causality measures with
	applications to nonlinear pulse-output network reconstruction}

\author[a,b,c]{Zhong-qi K. Tian}
\author[a,b,c]{Kai Chen} 
\author[a,b,c,1]{Songting Li} 
\author[d,e,f,g,h,1]{David W. McLaughlin}
\author[a,b,c,1]{Douglas Zhou}

\affil[a]{School of Mathematical Sciences, Shanghai Jiao Tong University, Shanghai 200240, China}
\affil[b]{Institute of Natural Sciences, Shanghai Jiao Tong University, Shanghai 200240, China}
\affil[c]{Ministry of Education Key Laboratory of Scientific and Engineering Computing, Shanghai Jiao Tong University, Shanghai 200240,
	China}

\affil[d]{Courant Institute of Mathematical Sciences, New York University, New York, NY 10012}
\affil[e]{Center for Neural Science, New York University, New York, NY 10012}
\affil[f]{Institute of Mathematical Sciences, New York University Shanghai, Shanghai 200122,
	China}
\affil[g]{New York University Tandon School of Engineering, New York University, Brooklyn, NY 11201}
\affil[h]{Neuroscience Institute of New York University
	Langone Health, New York University, New York, NY 10016}

\date{}

\maketitle

\begin{abstract}
The causal connectivity of a network is often inferred to understand the network function. It is arguably acknowledged that the inferred causal connectivity relies on causality measure one applies, and it may differ from the network's underlying structural connectivity. However, the interpretation of causal connectivity remains to be fully clarified, in particular, how causal connectivity depends on causality measures and how causal connectivity relates to structural connectivity.
Here, we focus on nonlinear networks with pulse signals as
measured output, $e.g.$, neural networks with spike output, and
address the above issues based on four intensively utilized causality measures, $i.e.$, time-delayed correlation, time-delayed
mutual information, Granger causality, and transfer entropy. We theoretically
show how these causality measures are related to one another when applied to pulse signals.  
Taking the simulated Hodgkin-Huxley neural network and the real mouse brain network as
two illustrative examples, we further verify the quantitative relations among the four causality measures and demonstrate that the causal connectivity inferred by any of the four well coincides with the underlying network structural connectivity,
therefore establishing a direct
link between the causal and structural connectivity. 
We stress that the structural connectivity of networks can be reconstructed pairwisely without conditioning on the global information
of all other nodes in a network, thus circumventing the curse of dimensionality. Our framework provides a practical and effective approach for pulse-output network reconstruction. 
\end{abstract}


\section*{Introduction}
The structural connectivity of a network is of great importance to
understand the cooperation and competition among nodes in the network.\cite{misic2016from}
However, it is often difficult to directly measure the structural
connectivity of a network, for instance, the brain network. 
On the other hand, with the development of experimental techniques, it becomes feasible to record high-temporal-resolution activities of nodes in a network simultaneously \cite{stringer2019spontaneous, wheeler2011in, stosiek2003vivo}.
This provides a possibility to reveal the underlying structural connectivity
by analyzing the nodes' activity data and identifying the causal interactions (connectivity) among nodes \cite{honey2009predicting, honey2010can, suarez2020linking}.
One of the most widely used statistical indicator for interaction identification is the correlation coefficient
\cite{benesty2009pearson,eisen1998cluster,ito2011extending} that
characterizes linear dependence between two nodes. Correlation coefficient is
symmetric thus cannot distinguish the driver-recipient relation to recover the
causal connectivity. To solve this, time-delayed correlation coefficient
(TDCC) \cite{bedenbaugh1997multiunit,ito2011extending} was introduced
being able to detect the direction of causal connectivity. However, TDCC, as a
linear measure, may fail to capture causal interactions in nonlinear
networks.  As a nonlinear generalization of TDCC, time-delayed
mutual information (TDMI) \cite{vastano1988information,schreiber2000measuring,frenzel2007partial}
as a model-free measure was proposed to measure the flow of information in nonlinear networks. 
Despite the mathematical simplicity and computational efficiency of TDCC and TDMI, as pointed out in previous works, both measures cannot exclude the historical effect of signals
and may encounter the issue of overestimation \cite{frenzel2007partial,schreiber2000measuring,mclaughlin2004use}.
Granger causality (GC) \cite{granger1969investigating,bressler2011wiener,barnett2018solved}
and transfer entropy (TE) \cite{schreiber2000measuring,bossomaier2016introduction,borge2016dynamics} were another
two measures to detect the causal connectivity  with the exclusion of
the historical effect of signal's own. GC is based on linear regression models that assumes
the causal relation can be revealed by analyzing low-order statistics (up
to the variance) of signals. Consequently, the validity of GC for nonlinear networks is in general questionable
\cite{li2018causal}. In contrast, TE is a nonparametric information-theoretic
measure that quantifies the causal interactions with no assumption of interaction models. However, it requires the estimation of the probability distribution of dynamical variables conditioning on the historical information in networks, which makes TE suffer 
from the curse of dimensionality in practical applications to network systems with many nodes \cite{runge2012escaping,newell2016mass,bach2017breaking}.

Despite the broad application of the above causality measures, to interpret the results,
two major theoretical issues remain to be clarified.
First,  what are the mathematical relations among these causality measures, $i.e.$, TDCC, TDMI, GC, and TE? 
It has been reported that the causal connectivity inferred by different causality measures in general can be inconsistent with one another \cite{marbach2010revealing,de2010advantages,zou2009granger}. 
Therefore, it is vital to understand the relation between different causal connectivity from those measures given certain conditions.
We note that TE has been proven to be equivalent to GC for Gaussian variables \cite{Barnett2009}, yet the mathematical relation among the four measures remains to be elucidated
for general variables. Second, what is the relation between the causal connectivity and the structural connectivity? 
Note that the causal connectivity inferred by these measures
is statistical rather than structural \cite{koch2002investigation,seth2005causal,schiele2013specific},
$i.e.$, the causal connectivity quantifies directed statistical correlation or dependence
among network nodes, whereas the structural connectivity corresponds
to physical connections among network nodes.  Therefore, it remains unclear whether the structural connectivity can be reconstructed from the causal connectivity in general. 

In this work, we address these questions by investigating a general class of  nonlinear networks with pulse
signals as measured output, which we term as pulse-output networks. We first reveal the relations
among TDCC, TDMI, GC, and TE with rigorous mathematical proofs. By taking the simulated Hodgkin-Huxley (HH) neural network and the real mouse brain network as two illustrative examples, we then verify the mathematical relations among the four causality measures numerically. We further demonstrate that the underlying
structural connectivity of these networks can be recovered from the causal connectivity
inferred using any of the above four causality measures.
We emphasize that the structural connectivity can be reconstructed pairwisely without conditioning on the 
global information of all other nodes, and thus circumvents the curse of dimensionality. 
Therefore, the reconstruction method based on these four causality measures can be 
applied to the reconstruction of structural connectivity in large-scale pulse-output nonlinear systems or subsystems. 

\section*{Results}

\subsection*{Concepts of generalized pairwise TDCC, TDMI, GC, and TE}

Consider a nonlinear network of $N$ nodes with dynamics given
by

\begin{equation}
\frac{d\mathbf{Z}}{dt}=\mathbf{F}(\mathbf{Z},t),\label{eq:general network}
\end{equation}
where $\mathbf{Z}=(Z_{1},Z_{2},...,Z_{N})$. We focus on the application of TDCC, TDMI, GC, and TE to each pair of nodes without conditioning on the rest of nodes in the network, accounting for the practical constraint that conditional causality measure in general requires the information of the whole network that is often difficult to observe. For the ease of illustration, we denote a pair of nodes as $X=Z_{i}$ and $Y=Z_{j}$, and their measured time series as $\{x_{n}\}$ and $\{y_{n}\}$, respectively.

TDCC \cite{bedenbaugh1997multiunit,ito2011extending},  as a function
of time delay $m$, is defined by 
\begin{equation*}
C(X,Y;m)=\frac{\text{cov}(x_{n},y_{n-m})}{\sigma_{x}\sigma_{y}},\label{eq:TDCC}
\end{equation*}
where ``cov'' represents the covariance, $\sigma_{x}$ and $\sigma_{y}$
are the standard deviations of $\{x_{n}\}$ and $\{y_{n}\}$, respectively.
A positive (negative) value of $m$ indicates the calculation of causal value
from $Y$ to $X$ (from $X$ to $Y$), and nonzero $C(X,Y;m)$ indicates the existence of causal interaction between $X$ and $Y$.
Without loss of generality, we consider the case of positive $m$ in the following discussions, that is, the causality measure from $Y$ to $X$.

In contrast to the linear measure TDCC, TDMI is a model-free method being able
to characterize nonlinear causal interactions \cite{vastano1988information,schreiber2000measuring,frenzel2007partial}.
TDMI from $Y$ to $X$ is defined by 
\begin{equation}
I(X,Y;m)=\sum_{x_{n},y_{n-m}}p(x_{n},y_{n-m})\log\frac{p(x_{n},y_{n-m})}{p(x_{n})p(y_{n-m})},\label{eq:MI}
\end{equation}
where $p(x_{n},y_{n-m})$ is the joint probability distribution of
$x_{n}$ and $y_{n-m}$, $p(x_{n})$ and $p(y_{n-m})$ are the corresponding
marginal probability distributions. $I(X,Y;m)$ is non-negative and
vanishes if and only if $x_{n}$ and $y_{n-m}$ are independent \cite{frenzel2007partial}.
Nonzero $I(X,Y;m)$ implies the existence of causal interaction
from $Y$ to $X$ for a positive $m$.

It has been noted that TDCC and TDMI could overestimate the causal interactions when a signal has a long memory  \cite{frenzel2007partial,schreiber2000measuring,mclaughlin2004use}.
As an alternative, GC was proposed to overcome the issue of overestimation based on linear regression
\cite{granger1969investigating,guo2008uncovering,bressler2011wiener}.
The auto-regression for $X$ is
represented by $x_{n+1}=a_{0}+\sum_{i=1}^{k}a_{i}x_{n+1-i}+\epsilon_{n+1},$
where $\{a_{i}\}$ are the auto-regression coefficients and $\epsilon_{n+1}$
is the residual. By including the historical information of $Y$ with a message
length $l$ and a time-delay $m$, the joint regression for $X$ is
represented by $x_{n+1}=\tilde{a}_{0}+\sum_{i=1}^{k}\tilde{a}_{i}x_{n+1-i}+\sum_{j=1}^{l}b_{j}y_{n+2-m-j}+\eta_{n+1},$
where $\{\tilde{a}_{i}\}$ and $\{b_{j}\}$ are the joint regression
coefficients, and $\eta_{n+1}$ is the corresponding residual. 
If there exists a causal interaction from $Y$ to $X$, then the prediction of $X$ using the linear regression models shall be improved by additionally incorporating the historical information of $Y$. Accordingly, the variance of residual $\eta_{n+1}$ is smaller than that of $\epsilon_{n+1}$. Based on this concept, the
GC value from $Y$ to $X$ is defined by 

\begin{equation*}
G_{Y\rightarrow X}(k,l;m)=\log\frac{\textrm{Var}(\epsilon_{n+1})}{\textrm{Var}(\eta_{n+1})}.\label{eq:GC}
\end{equation*}
The GC value is also non-negative and vanishes if and only
if $\{b_{j}\}=0$, $i.e.$, the variance of residual $\epsilon_{n+1}$ for
$X$ cannot be reduced by including the historical information of $Y$. 
Note that, by introducing the time-delay parameter $m$, the GC analysis defined above generalizes the conventional GC analysis, as the latter corresponds to the special case of $m=1$. 

GC assumes that the causal interaction can be fully captured by the variance reduction in the linear regression models, which is valid for Gaussian signals but not for more general signals.
As a nonlinear extension of GC, TE was proposed to describe the causal interaction
from the information theoretic perspective \cite{schreiber2000measuring}.
The TE value from $Y$ to $X$ is defined by

\begin{equation}
\begin{aligned}
T_{Y\rightarrow X}(k,l;m) =&\sum_{x_{n+1},x_{n}^{(k)},y_{n+1-m}^{(l)}}p(x_{n+1},x_{n}^{(k)},y_{n+1-m}^{(l)})\\
& \cdot\log\frac{p(x_{n+1}|x_{n}^{(k)},y_{n+1-m}^{(l)})}{p(x_{n+1}|x_{n}^{(k)})},
\end{aligned}
\label{eq:TE}
\end{equation}
where the shorthand notation $x_{n}^{(k)}=(x_{n},x_{n-1},...,x_{n-k+1})$
and $y_{n+1-m}^{(l)}=(y_{n+1-m},y_{n-m},...,y_{n+2-m-l})$, $k,l$
indicate the length (order) of historical information of $X$ and $Y$,
respectively. Similar to GC, the time-delay parameter $m$ is introduced that generalizes the conventional TE, the latter of which corresponds to the case of $m=1$.
TE is non-negative
and vanishes if and only if $p(x_{n+1}|x_{n}^{(k)},y_{n+1-m}^{(l)})=p(x_{n+1}|x_{n}^{(k)})$,
$i.e.,$ the uncertainty of $x_{n+1}$ is not affected regardless
of whether the historical information of $Y$ is taken into account.

In this work, we investigate the mathematical relations among TDCC, TDMI, GC, and TE by focusing on nonlinear networks described by Eq. \ref{eq:general network} with pulse signals as measured output, $e.g.$, the spike trains measured in neural networks. 
Consider a pair of nodes $X$ and $Y$ in the network of
$N$ nodes, and denote their
pulse-output signals by 
\begin{equation}
w_{x}(t)=\sum_{l}\delta(t-\tau_{xl}) \,\,\,\, \text{ and } \,\,\,\, w_{y}(t)=\sum_{l}\delta(t-\tau_{yl}),
\label{eq:setup}
\end{equation}
respectively, where $\delta(\cdot)$ is the Dirac delta function, and $\{\tau_{xl}\}$ and $\{\tau_{yl}\}$ are the output time sequences of nodes $X$ and $Y$ determined by Eq. \ref{eq:general network}, respectively. With the sampling resolution of $\Delta t$, 
the pulse-output signals are measured as binary time series $\{x_{n}\}$
and $\{y_{n}\}$, where $x_n=1$ ($y_n=1$) if there is a pulse signal, $e.g.$, a spike generated by a neuron, 
of $X$ ($Y$) occurred in the time window $[t_n,t_{n}+\Delta t)$, and $x_n=0$ ($y_n=0$) otherwise, $i.e.$,

\begin{equation}\label{eq:setup1}
x_n=\int_{t_n}^{t_{n}+\Delta t}w_x(t)dt \quad \text{and} \quad  y_n=\int_{t_n}^{t_{n}+\Delta t}w_y(t)dt,
\end{equation}
and $t_n=n\Delta t$. 
Note that the value of $\Delta t$ is often chosen to make sure that there is at most one pulse signal in a single time window.   
In the stationary resting state, the responses $x_n$ and $y_n$ can be viewed as stochastic processes when the network is driven by stochastic external inputs. 
In such a case, for the sake of simplicity, we denote $p_x=p(x_n=1)$, $p_y=p(y_n=1)$, and $\Delta p_{m}=\frac{p(x_{n} = 1,y_{n-m} = 1)}{p(x_{n} = 1 )p(y_{n-m} = 1)}-1$
that measures the dependence between $x_{n}$ and $y_{n-m}$. 

\subsection*{Mathematical relation between TDMI and TDCC}

For the relation between TDCC and TDMI when applied to the above nonlinear networks with pulse-output signals, we prove the following theorem:
\begin{thm}
	For nodes $X$ and $Y$ with pulse-output signals given in Eqs. \ref{eq:setup} and \ref{eq:setup1}, we have 
	
	\begin{equation}
	I(X,Y;m)=\frac{C^{2}(X,Y;m)}{2}+O(\Delta t^{2}\Delta p_{m}^{3}),\label{eq:MI vs CC}
	\end{equation}
	where the symbol $``O"$ stands for the order. \label{thm: MI vs CC}
\end{thm}

\begin{proof}
	The basic idea is to Taylor expand TDMI in Eq. \ref{eq:MI} with respect to the term $\frac{p(x_{n},y_{n-m})}{p(x_{n})p(y_{n-m})}-1$ (the detailed derivation can be found in \textit{\textcolor{blue}{SI Appendix, Supplementary Information Text 1B}}), then we arrive at the following expression:	
	
	
	\begin{equation*}
	\begin{aligned} 
	I(X,Y;m)&=\sum_{x_{n},y_{n-m}}p(x_{n},y_{n-m})\log\!\left(\!1+\frac{p(x_{n},y_{n-m})}{p(x_{n})p(y_{n-m})}-1\!\right)\\
	& =\frac{\left[p(x_{n}=1,y_{n-m}=1)-p_{x}p_{y}\right]{}^{2}}{2(p_{x}-p_{x}^{2})(p_{y}-p_{y}^{2})}+O(\Delta t^{2}\Delta p_{m}^{3}).
	\end{aligned}
	\label{eq:derivation MI CC 1}
	\end{equation*}
	By definition,
	\begin{equation}
	C(X,Y;m)=\frac{p(x_{n}=1,y_{n-m}=1)-p_{x}p_{y}}{\sqrt{(p_{x}-p_{x}^{2})(p_{y}-p_{y}^{2})}},\label{eq:derivaiton MI CC 2}
	\end{equation}
	we have
	
	\begin{equation*}
	I(X,Y;m)=\frac{C^{2}(X,Y;m)}{2}+O(\Delta t^{2}\Delta p_{m}^{3}).
	\end{equation*}
\end{proof}

\subsection*{Mathematical relation between GC and TDCC}

%

We next derive the relation between GC and TDCC as follows:
\begin{thm} \label{thm: GC vs CC}
	For nodes $X$ and $Y$ with pulse-output signals given in Eqs. \ref{eq:setup} and \ref{eq:setup1}, we have 
	\begin{equation}
	G_{Y\rightarrow X}(k,l;m)=\sum_{i=m}^{m+l-1}C^{2}(X,Y;i)+O(\Delta t^{3}\Delta p_{m}^{2}).\label{eq:GC vs CC}
	\end{equation}
\end{thm}

\begin{proof}
	From the definition, GC can be represented by the covariances of the signals \cite{Barnett2009} as
	\begin{equation}
	G_{Y\rightarrow X}(k,l;m) =\log\frac{\Gamma(x_{n+1}|x_{n}^{(k)})}{\Gamma(x_{n+1}|x_{n}^{(k)}\oplus y_{n+1-m}^{(l)})},
	\label{eq:GC comp}
	\end{equation}
	where $\Gamma(\mathbf{x}|\mathbf{y})=\textrm{cov}(\mathbf{x})-\textrm{cov}(\mathbf{x},\mathbf{y})\textrm{cov}(\mathbf{y})^{-1}\textrm{cov}(\mathbf{x},\mathbf{y})^{T}$
	for random vectors $\mathbf{x}$ and $\mathbf{y}$, $\textrm{cov}(\mathbf{x})$ and $\textrm{cov}(\mathbf{y})$
	denote the covariance matrix of $\mathbf{x}$ and $\mathbf{y}$, respectively, and $\textrm{cov}(\mathbf{x},\mathbf{y})$
	denote the cross-covariance matrix between $\mathbf{x}$ and $\mathbf{y}$. The symbol $T$ is the transpose
	operator and $\oplus$ denotes the concatenation of vectors. 
	
	We first prove that the auto-correlation function (ACF) of binary time series $\{x_n\}$ as a function of time delay takes the order of $\Delta t$ 
	(see \textit{\textcolor{blue}{SI Appendix, Supplementary Information Text 1C}} for the proof, \textcolor{blue}{Fig. S1}). Accordingly, we have
	
	\begin{equation*}
	\textrm{cov}(x_{n}^{(k)})  =\sigma_{x}^{2}(\mathbf{I}+\hat{\mathbf{A}}),
	\end{equation*}
	where $\hat{\mathbf{A}}=(\hat{a}_{ij}),\hat{a}_{ij}=O(\Delta t)$, and $\mathbf{I}$ is the identity matrix. 
	Hence,
	\begin{equation} \label{eq:derivation GC CC 1}
	\begin{aligned} 
	\Gamma(x_{n+1}|x_{n}^{(k)})=&\sigma_{x}^{2}-\frac{1}{\sigma_{x}^{2}}\textrm{cov}(x_{n+1},x_{n}^{(k)})(\mathbf{I}-\hat{\mathbf{A}})\\
	& \cdot\textrm{cov}(x_{n+1},x_{n}^{(k)})^{T}+O(\Delta t^{5}).
	\end{aligned}
	\end{equation}
	In the same way, we have
	\begin{equation}
	\begin{aligned} 
	\Gamma(&x_{n+1}|x_{n}^{(k)}\oplus y_{n+1-m}^{(l)})\\
	= &\sigma_{x}^{2}-\frac{1}{\sigma_{x}^{2}}\textrm{cov}(x_{n+1},x_{n}^{(k)})(\mathbf{I}-\hat{\mathbf{A}})\textrm{cov}(x_{n+1},x_{n}^{(k)})^{T}\\
	& -\frac{1}{\sigma_{y}^{2}}\textrm{cov}(x_{n+1},y_{n+1-m}^{(l)})(\mathbf{I}-\hat{\mathbf{B}})\textrm{cov}(x_{n+1},y_{n+1-m}^{(l)})^{T}\\
	& +O(\Delta t^{4}\Delta p_{m}^{2}).
	\end{aligned}
	\label{eq:derivation GC CC 2}
	\end{equation}
	where $\hat{\mathbf{B}}=(\hat{b}_{ij}),\hat{b}_{ij}=O(\Delta t)$.  
	Substituting Eqs. \ref{eq:derivation GC CC 1} and \ref{eq:derivation GC CC 2} into Eq. \ref{eq:GC comp} and Taylor expanding Eq. \ref{eq:GC comp} with respect to $\Delta t$, we can obtain
	\begin{equation}
	G_{Y\rightarrow X}(k,l;m)=\sum_{i=m}^{m+l-1}C^{2}(X,Y;i)+O(\Delta t^{3}\Delta p_{m}^{2}).\label{eq:derivation GC CC 3}
	\end{equation}
	The detailed derivation of Eqs. \ref{eq:derivation GC CC 1}, \ref{eq:derivation GC CC 2},
	and \ref{eq:derivation GC CC 3} can be found in \textit{\textcolor{blue}{SI Appendix, Supplementary Information Text 1C.}} 
\end{proof}

\subsection*{Mathematical relation between TE and TDMI}
From the generalized definitions of TE and TDMI, TE can be regarded as a generalization of TDMI conditioning on the signals' historical information additionally. 
To rigorously establish their relationship,
we require that $\Bigl\Vert x_{n+1}^{(k+1)}\Bigr\Vert_{0}\leq1$ and $\Bigl\Vert y_{n+1-m}^{(l)}\Bigr\Vert_{0}\leq1$ in the definition of TE given in Eq. \ref{eq:TE},
where $\Bigl\Vert \cdot \Bigr\Vert_{0}$ denotes the $l_{0}$ norm of a vector, $i.e.$, the number of nonzero elements in a vector. 
This assumption indicates that the length of historical information used in the TE framework is shorter than the minimal time interval between two consecutive pulse-output signals. 
This condition is often satisfied for nonlinear network with output signals. 
Consequently, we have the following theorem:
\begin{thm}  \label{thm: TE vs TDMI}
	For nodes $X$ and $Y$ with pulse-output signals given in Eqs. \ref{eq:setup} and \ref{eq:setup1}, we have
	\begin{equation}
	T_{Y\rightarrow X}(k,l;m)=\sum_{i=m}^{m+l-1}I(X,Y;i)+O(\Delta t^{3}\Delta p_{m}^{2}),\label{eq:TE vs MI}
	\end{equation}
	where $T_{Y\rightarrow X}$ is defined in Eq. \ref{eq:TE} with the assumption that $\Bigl\Vert x_{n+1}^{(k+1)}\Bigr\Vert_{0}\leq1$ and $\Bigl\Vert y_{n+1-m}^{(l)}\Bigr\Vert_{0}\leq1$.
\end{thm}

\begin{proof}
	To simplify the notations, we denote $x_{n}^{(k)}=(x_{n},x_{n-1},...,x_{n-k+1})$
	and $y_{n+1-m}^{(l)}=(y_{n+1-m},y_{n-m},...,y_{n+2-m-l})$ as $x^{-}$ and $y^{-}$, respectively.
	From Eq. \ref{eq:TE}, we have
	\begin{equation}
	\begin{aligned} 
	T_{Y\rightarrow X}(k,l;m)\!&=\!\!\sum_{x_{n+1},x^{-},y^{-}}\!p(x_{n+1},x^{-},y^{-})\log\frac{p(x_{n+1}|x^{-},y^{-})}{p(x_{n+1}|x^{-})}\\
	& =\sum_{i=m}^{m+l-1}I(X,Y;i)+\mathcal{A}+\mathcal{B},
	\end{aligned}
	\label{eq:derivation TE MI}
	\end{equation}
	where \[\mathcal{A}=\sum_{x_{n+1},y^{-}}p(x_{n+1},y^{-})\log\frac{p(y^{-}|x_{n+1})}{\prod_{j}p(y_{j}|x_{n+1})}\frac{\prod_{j}p(y_{j})}{p(y^{-})}\]
	and \[\mathcal{B}=\!\!\sum_{x_{n+1},x^{-},y^{-}}p(x_{n+1},x^{-},y^{-})\log\frac{p(x_{n+1}|x^{-},y^{-})}{p(x_{n+1}|y^{-})}\frac{p(x_{n+1})}{p(x_{n+1}|x^{-})}\]		
	and $\prod_{j}$ in $\mathcal{A}$ represents $\prod_{j=n+2-m-l}^{n+1-m}$. The detailed derivation of Eq. \ref{eq:derivation TE MI} can be found in \textit{\textcolor{blue}{SI Appendix, Supplementary Information Text 1D.}} 
	$\mathcal{A}$ and $\mathcal{B}$ in Eq. \ref{eq:derivation TE MI}
	consist of multiple terms and the leading order of each term can be analytically calculated. 
	For the sake of illustration, below we derive the leading order of one of these terms and the leading order of the rest terms can be estimated in a similar way.
	Under the assumption that $\Bigl\Vert x_{n+1}^{(k+1)}\Bigr\Vert_{0}\leq1$ and $\Bigl\Vert y_{n+1-m}^{(l)}\Bigr\Vert_{0}\leq1$, the number of nonzero components is at most one in $x_{n+1}^{(k+1)}$ and $y_{n+1-m}^{(l)}$. 
	Without loss of generality, we assumed $x_{n+1}=1$ and $y_{n+1-m}=1$. In such a case, we can obtain the following expression in $\mathcal{A}$,
	
	\begin{equation*}
	\begin{aligned}
	p(x_{n+1},y^{-})&\log\frac{p(y^{-}|x_{n+1})}{\prod_{j}p(y_{j}|x_{n+1})}\frac{\prod_{j}p(y_{j})}{p(y^{-})}\biggl|_{x_{n+1}=1,y_{n+1-m}=1}\\
	& =p_{x}p_{y}^{2}\sum_{j\neq n+1-m}\Delta p_{n+1-j}+O(\Delta t^{3}\Delta p_{m}^{2}).
	\end{aligned}
	\end{equation*}
	We can further show that the leading order of all the terms in $\mathcal{A}$  happen to cancel each other out (\textit{\textcolor{blue}{SI Appendix, Supplementary Information Text 1D}}), thus we have $\mathcal{A}=O(\Delta t^{3}\Delta p_{m}^{2})$.
	Similarly, we can also show $\mathcal{B}=O(\Delta t^{3}\Delta p_{m}^{2})$
	(\textit{\textcolor{blue}{SI Appendix, Supplementary Information Text 1D}}), and thus 
	
	\begin{equation*}
	T_{Y\rightarrow X}(k,l;m)=\sum_{i=m}^{m+l-1}I(X,Y;i)+O(\Delta t^{3}\Delta p_{m}^{2}).
	\end{equation*}
\end{proof}

\subsection*{Mathematical relation between GC and TE}
From Theorems \ref{thm: MI vs CC}-\ref{thm: TE vs TDMI}, we can straightforwardly prove the following theorem:

\begin{thm}  \label{thm: TE vs GC}
	For nodes $X$ and $Y$ with pulse-output  signals given in Eqs. \ref{eq:setup} and \ref{eq:setup1}, we have
	\begin{equation}
	G_{Y\rightarrow X}(k,l;m)=2T_{Y\rightarrow X}(k,l;m)+O(\Delta t^{2}\Delta p_{m}^{3}),\label{eq:TE vs GC}
	\end{equation}
	where $T_{Y\rightarrow X}$ is defined in Eq. \ref{eq:TE} with the assumption that $\Bigl\Vert x_{n+1}^{(k+1)}\Bigr\Vert_{0}\leq1$ and $\Bigl\Vert y_{n+1-m}^{(l)}\Bigr\Vert_{0}\leq1$.
\end{thm}

Note that in Theorems \ref{thm: TE vs TDMI} and \ref{thm: TE vs GC} we have the extra assumption $\Bigl\Vert x_{n+1}^{(k+1)}\Bigr\Vert_{0}\leq1$ and $\Bigl\Vert y_{n+1-m}^{(l)}\Bigr\Vert_{0}\leq1$. As numerically verified below, the extra assumption is often easily satisfied to achieve the successful application of these causality measures. For example, the memory time of the neuron is about 20 ms, while the inter-spike interval is around 100 ms. The extra assumption is to rigorously establish the relations in Theorems \ref{thm: TE vs TDMI} and \ref{thm: TE vs GC}, but is not a necessary condition (\textit{\textcolor{blue}{SI Appendix}}, \textcolor{blue}{Fig. S2}). 
We summarize the relations among the four causality measures in Fig. \ref{fig:TGIC relation}.

\begin{figure}[!htp]
	\begin{centering}
		\includegraphics[width=0.5\linewidth]{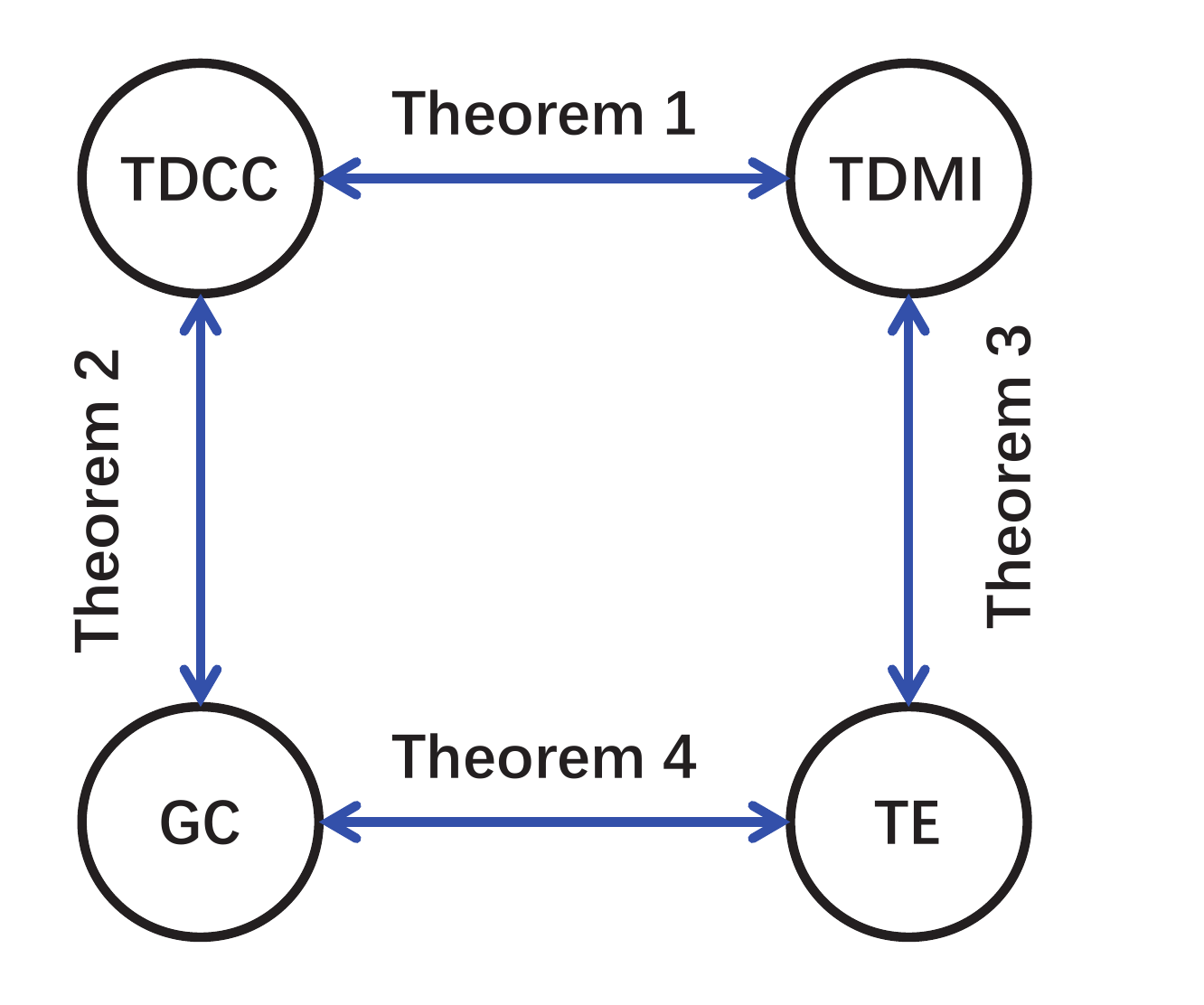}
		\par\end{centering}
	\caption{Mathematical relations among TDCC, TDMI, GC, and TE.
		\label{fig:TGIC relation} }
\end{figure}

\subsection*{Mathematical relations of causality measures verified in HH neural networks}
To verify the relations among the causality measures derived above, as an illustrative example,
we apply generalized pairwise TDCC, TDMI, GC, and TE to the HH neural network described in \hyperref[sec:Materials-and-Methods]{\textit{Materials and Methods}}. 
We first consider a pair of neurons denoted by $X$ and $Y$ with unidirectional connection from $Y$ to $X$ in an HH network containing 10 excitatory neurons driven by homogeneous Poisson inputs.
Let $\{\tau_{xl}\}$ and $\{\tau_{yl}\}$ be the ordered spike times of neuron $X$ and $Y$ in the HH network respectively and denote their spike trains as $w_{x}(t)=\sum_{l}\delta(t-\tau_{xl})$ and $w_{y}(t)=\sum_{l}\delta(t-\tau_{yl})$, respectively.
With a sampling resolution of $\Delta t$, the spike train is measured as a binary time series as described above.
To numerically verify the above theorems, we then check the order of the remainders in Eqs. \ref{eq:MI vs CC},
\ref{eq:GC vs CC}, \ref{eq:TE vs MI}, and \ref{eq:TE vs GC} in terms of $\Delta t$ and $\Delta p_{m}$.
Note that $\Delta p_{m}$, the measure of the dependence between $X$ and $Y$, is insensitive to sampling resolution $\Delta t$ (\textit{\textcolor{blue}{SI Appendix}}, \textcolor{blue}{Fig. S3}).
Therefore, by varying sampling interval $\Delta t$ and coupling strength $S$ (linearly related to $\Delta p_m$) respectively, the orders of the remainders are consistent with those derived in Eqs. \ref{eq:MI vs CC}, \ref{eq:GC vs CC}, \ref{eq:TE vs MI}, and \ref{eq:TE vs GC} (Fig. \ref{fig:conv test}).
In addition, Fig. \ref{fig:parameters} verifies the relations among the causality measures by changing other parameters. For example, in Eqs. \ref{eq:GC vs CC}, \ref{eq:TE vs MI}, and \ref{eq:TE vs GC}, the four causality measures are proved to be independent of the historical length $k$, which is numerically verified in Fig. \ref{fig:parameters}\textit{A}. 
And although the values of GC and TE rely on the historical length $l$ (Fig. \ref{fig:parameters}\textit{B}), the mathematical relations among the four causality measures revealed by Theorems \ref{thm: MI vs CC}-\ref{thm: TE vs GC} hold for a wide range of $l$. 
We next verify the mathematical relations among the causality measures for the parameter of time delay $m$ by fixing parameters $k$ and $l$. In principle, the value of $k$ and $l$ in GC and TE shall be determined by the historical memory of the system. To reduce the computational cost
\cite{gourevitch2007evaluating,vicente2011transfer,ito2011extending}, we take $k=l=1$ for all the results below. It turns out that this parameter choice works well for pulse-output networks because of the short memory effect in general, as will be further discussed later. 
Fig. \ref{fig:parameters}\textit{C} shows the mathematical relations hold for a wide range of time delay parameter used in computing the four causality measures.

\begin{figure}[!htp]
	\begin{centering}
		\includegraphics[width=1\linewidth]{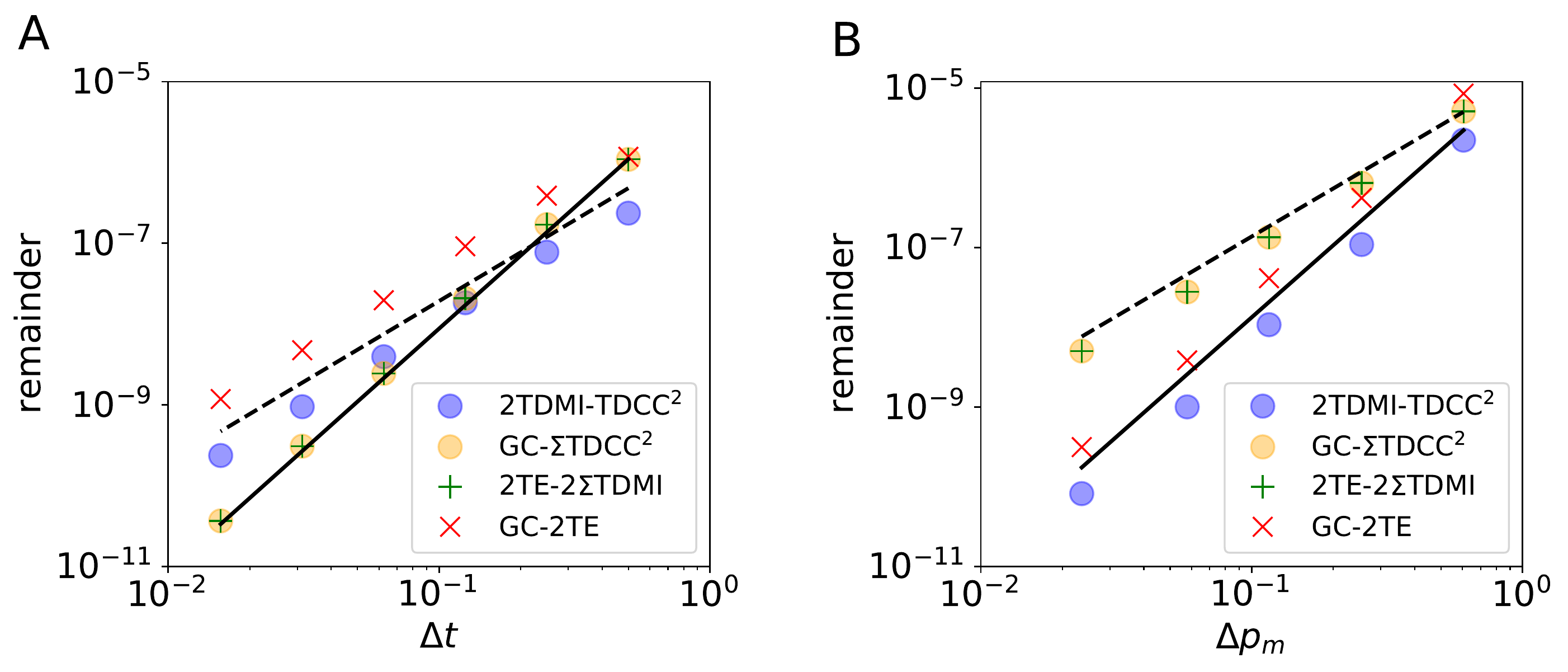}
		\par\end{centering}
	\caption{Convergence order of the remainders in terms of (\textit{A}) $\Delta t$ and (\textit{B}) $\Delta p_{m}$
		obtained from causality measures from neuron $Y$ to neuron $X$ with unidirectional connection from $Y$ to $X$ in an HH network of 10 excitatory neurons randomly connected with probability 0.25. The convergence orders for $\Delta t$ in (\textit{A}) fit with which shows in Theorems \ref{thm: MI vs CC}-\ref{thm: TE vs GC} ($R^2>0.998$), while those for $\Delta p_m$ in (\textit{B}) fit the values derived in Theorems \ref{thm: MI vs CC}-\ref{thm: TE vs GC} ($R^2>0.998$). The gray dashed and solid lines indicate the $2^\mathrm{nd}$-order and $3^\mathrm{rd}$-order convergence, respectively.
		The parameters are set as $k=l=5$ and $m=6$ (time delay is 3 ms), $S=0.02$ $\textrm{mS\ensuremath{\cdot}cm}^{-2}$ in (\textit{A}), and $\Delta t=0.5$ ms in (\textit{B}). 
		\label{fig:conv test} }
\end{figure}

We further examin the robustness of the mathematical relations among
TDCC, TDMI, GC, and TE by scanning the parameters of the coupling strength $S$ between the HH neurons and external Poisson input strength and rates. As shown in Fig. \ref{fig:parameters}\textit{D}, the values of the four causality measures with different coupling strength are very close to one another. Their relations also hold for a wide range of external Poisson input parameters (\textit{\textcolor{blue}{SI Appendix}},
\textcolor{blue}{Fig. S4}).
From the above, the mathematical relations among TDCC, TDMI, GC, and TE described in Theorems \ref{thm: MI vs CC}-\ref{thm: TE vs GC} are verified in the HH network.

\begin{figure}[!htp]
	\begin{centering}
		\includegraphics[width=1\linewidth]{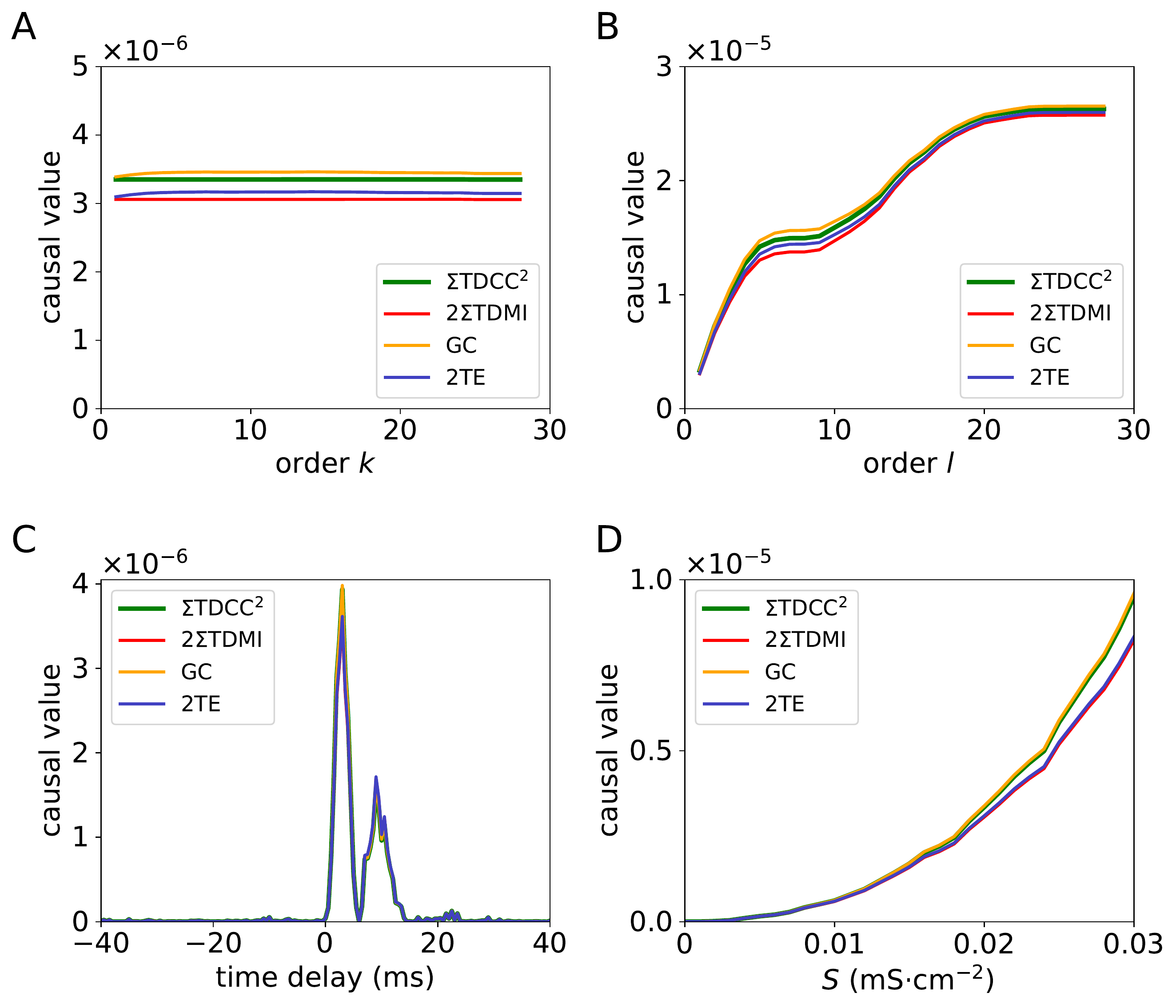}
		\par\end{centering}
	\caption{Dependence of causal values on parameters of (\textit{A}) order $k$, (\textit{B}) order $l$, (\textit{C}) time delay, and (\textit{D}) coupling strength $S$ obtained from neuron $Y$ to
		neuron $X$ in the HH network in Fig. \ref{fig:conv test}. 
		In (\textit{C}), a positive (negative) time delay indicates the calculation of causal values from $Y$ to $X$ (from $X$ to $Y$). 
		The curves of causal values reach the peak at the time delay 3 ms ($m=6$) in (\textit{C}).
		The green curve represents the summation of square of TDCC $C(X,Y;m)$, the red curve represents two times the summation of TDMI $I(X,Y;m)$, the orange curve stands for GC $G_{Y\rightarrow X}(k,l;m)$, and the blue curve stands for twice of TE $T_{Y\rightarrow X}(k,l;m)$. The curves are close to each other in (A)-(D) and virtually overlap (all significantly greater than those of randomly surrogate time series, $p<0.05$).
		The parameters are set as 
		(\textit{A}): $l=1$, $S=0.02$ $\textrm{mS\ensuremath{\cdot}cm}^{-2}$, $\Delta t=0.5$ ms, $m=6$; 
		(\textit{B}): $k=1$, $S=0.02$ $\textrm{mS\ensuremath{\cdot}cm}^{-2}$, $\Delta t=0.5$ ms, $m=6$;
		(\textit{C}): $k=l=1$, $S=0.02$ $\textrm{mS\ensuremath{\cdot}cm}^{-2}$, $\Delta t=0.5$ ms; 
		(\textit{D}): $k=l=1$, $\Delta t=0.5$ ms, $m=6$.		 
		\label{fig:parameters}}
\end{figure} 

\subsection*{Relation between structural connectivity and causal connectivity in HH neural networks}
We next discuss about the relation between the inferred causal connectivity and the structural connectivity. Note that the causal connectivity inferred by these measures are statistical rather than structural \cite{koch2002investigation,seth2005causal,schiele2013specific}, $i.e.$, the causal connectivity quantifies the directed statistical correlation or dependence among network nodes, whereas the structural connectivity corresponds to physical connections among network nodes. 
Therefore, it remains unclear about the relationship between the causal connectivity and the structural connectivity.
In Fig. \ref{fig:parameters}\textit{C}, the peak causal value from $Y$ to $X$ (at time delay around 3 ms, m=6) is significantly greater than threshold, while the causal value from $X$ to $Y$ is not. Based on this, the inferred causal connections between $X$ and $Y$ are consistent with the underlying structural connections. From now on, we adopt peak causal values, m=6, to represent the causal connectivity unless noted explicitly.
To investigate the validity of this consistency in larger networks, we further investigate a larger HH network (100 excitatory neurons) with random connectivity (\textit{\textcolor{blue}{SI Appendix}}, \textcolor{blue}{Fig. S5A}). As shown in Fig. \ref{fig:distribution}\textit{A},
the distributions of all four causal values across all pairs of neurons virtually overlapped, which again verify their mathematical relations given by Theorems \ref{thm: MI vs CC}-\ref{thm: TE vs GC}.
In addition, as the network size increases, the distributions of the causality measures in Fig. \ref{fig:distribution}\textit{A} exhibits a bimodal structure with a clear separation of orders. By mapping the causal values with the structural connectivity, we find that the right bump of the distributions with larger causal values corresponded to connected pairs of neurons, while the left bump with smaller causal values corresponded to unconnected pairs. 
The well separation of the two modals indicates that the underlying structural connectivity in the HH network can be accurately reconstructed by the causal connectivity. The performance of this reconstruction approach can be quantitatively characterized by the receiver operating characteristic (ROC) curve and the area under the ROC curve (AUC) \cite{fawcett2006introduction,marbach2010revealing,carter2016roc} (\hyperref[sec:Materials-and-Methods]{\textit{Materials and Methods}}). It is found that the AUC value became 1 when applying any of these four causality measures (\textit{\textcolor{blue}{SI Appendix}}, \textcolor{blue}{Fig. S5B}), which indicates that the structural connectivity of the HH network could be reconstructed with 100 \% accuracy. 
We point out that the reconstruction of network connectivity based on causality measures is achieved by calculating the causal values between each pair of neurons that required no access towards the activity data of the rest of neurons. Therefore, this inference approach can be applied to a subnetwork when the activity of neurons outside the subnetwork is not observable. 
For example, when a subnetwork of 20 excitatory HH neurons is observed, the structural connectivity
of the subnetwork can still be accurately reconstructed without knowing the information
of the rest 80 neurons in the full network (Fig. \ref{fig:distribution}\textit{B}). In such a case, 
the AUC values corresponding to the four causality measures are 1 (\textit{\textcolor{blue}{SI Appendix}}, \textcolor{blue}{Fig. S5C}). 

\begin{figure}[!htp]
	\begin{centering}
		\includegraphics[width=1\linewidth]{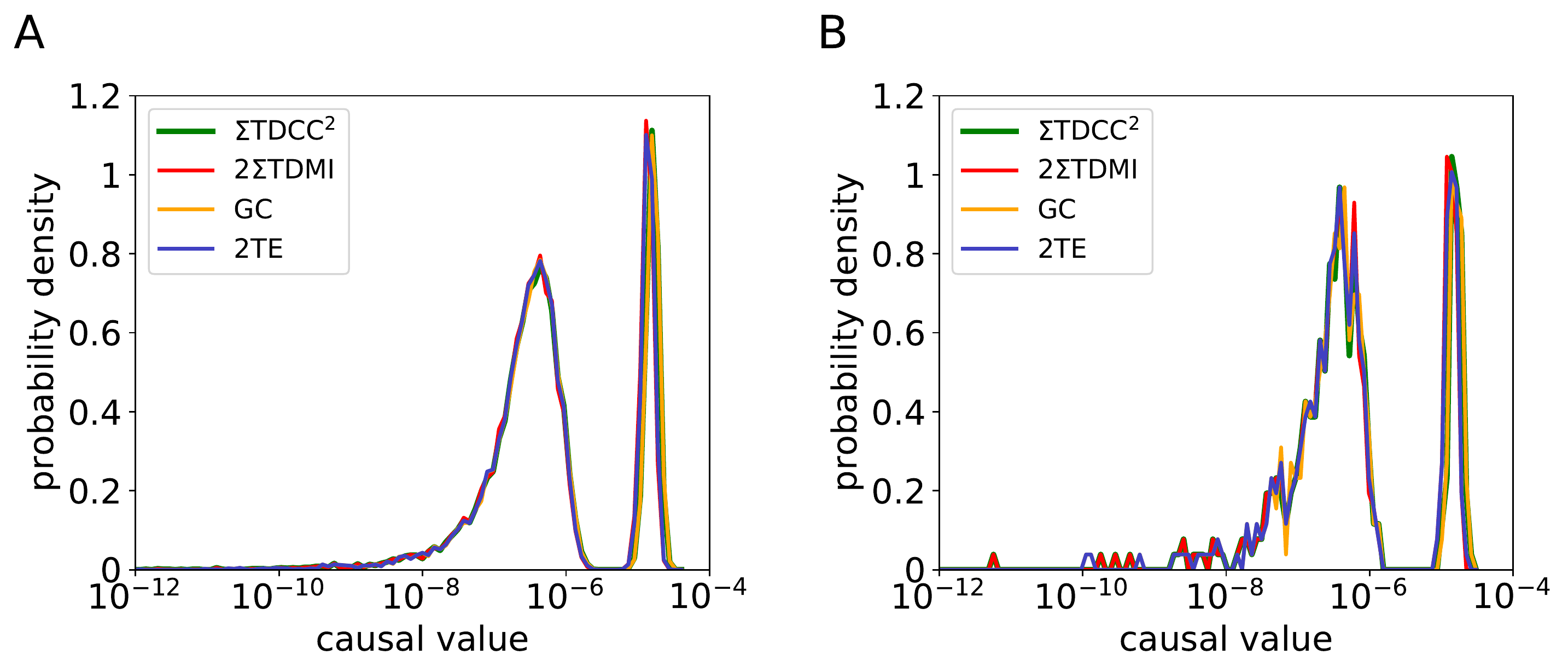}
		\par\end{centering}
	\caption{Distributions of causal values in an HH network of 100 excitatory neurons randomly connected with probability 0.25.
		(\textit{A}) The distribution  of causal values
		of each pair of neurons in the whole network. (\textit{B}) The distribution of causal values
		of each pair of neurons in an HH subnetwork of 20 excitatory neurons.
		The parameters are set as  $k=l=1$, and $S=0.02$
		$\textrm{mS\ensuremath{\cdot}cm}^{-2}$, $\Delta t=0.5$ ms, and $m=6$ (3 ms). The colors are the same as those in Fig. \ref{fig:parameters} and the curves nearly overlap. \label{fig:distribution}}
\end{figure}

\subsection*{Mechanism underlying network connectivity reconstruction by causality measures}
We next demonstrate the mechanism of pairwise inference on pulse-output signals in the reconstruction of network structural connectivity. 
It has been noticed that pairwise causal inference may potentially fail to distinguish the direct interactions from the
indirect ones in a network. For example, in the case that  $Y\rightarrow W\rightarrow X$ where ``$\rightarrow$'' denotes a directed connection,
the indirect interaction from $Y$ to $X$ may possibly be mis-inferred as a direct interaction via pairwise causality measures. However, such spurious inference can be circumvented in our generalized pairwise causality measured based on pulse-output signals as explained below. 
Here we take TDCC as an example to elucidate the underlying reason of successful reconstruction. If we denote $\delta{p_{Y\rightarrow X}}=p(x_{n}=1|y_{n-m}=1)-p(x_{n}=1|y_{n-m}=0)$ as the increment of probability of generating a pulse output for neuron $X$ induced by a pulse-output signal of neuron $Y$ at $m$ time step earlier, we have TDCC $C(X,Y;m)=\delta p_{Y\rightarrow X}\sqrt{\frac{p_{y}-p_{y}^{2}}{p_{x}-p_{x}^{2}}}$ through Eq. \ref{eq:derivaiton MI CC 2}.
For the case of $Y\rightarrow W\rightarrow X$, we can derive $\delta{p}_{Y\rightarrow X}=O(\delta{p}_{Y\rightarrow W}\cdot\delta{p}_{W\rightarrow X})$ (details in \textit{\textcolor{blue}{SI Appendix, Supplementary Information Text 2}}). 
Because the influence of a single pulse output signal is often small ($e.g.$, in the HH neural network with physiologically realistic coupling strengths, we obtain $|\delta{p}|<0.01$ from simulation data), the causal value $C(X,Y;m)$ due to the indirect interaction is significantly smaller than $C(W,Y;m)$ or $C(X,W;m)$ by direct interactions(Fig. \ref{fig:distribution}\textit{A}). 

We also note that the increment $\delta{p}$ is linearly dependent on the coupling strength $S$ (\textit{\textcolor{blue}{SI Appendix}}, \textcolor{blue}{Fig. S6}),
thus we establish a mapping between the causal and structural connectivity, in which the causal value of TDCC is proportional to the coupling strength $S$ between two neurons. The mapping between causal and structural connectivity for TDMI, GC, and TE are also established in a similar way, in which the corresponding causal values are proportional to $S^2$. 
Therefore, the application of pairwise causality measures to pulse-output signals is able to successfully reveal the underlying structural connectivity of a network. Importantly, this approach overcomed the computational issue of high dimensionality, thus is potentially applicable to data of large-scale networks or subnetworks measured in experiments.


\subsection*{Network connectivity reconstruction with physiological experimental data}
Next, we apply all four causality measures to experimental data to address the issue of validity of their mathematical relations and reconstruction of the network structural connectivity. Here, we analyze the in vivo spike data recorded in the mouse brain from Allen Brain Observatory \cite{allen2016allen} (details in \hyperref[sec:Materials-and-Methods]{\textit{Materials and Methods}}). 
By applying the four causality measures, we infer the causal connectivity of those brain networks. As the underlying structural connectivity of the recorded neurons in experiments was unknown, we first detect putative connected links from the distribution of causality measures as introduced below, and then follow the same procedures as previously described in the HH model case using ROC and AUC in signal detection theory \cite{fawcett2006introduction,marbach2010revealing,carter2016roc} (\hyperref[sec:Materials-and-Methods]{\textit{Materials and Methods}}) to quantify the reconstruction performance. 

Since demonstrated the equivalence of four causality measures, we choose TE as an example for the discussion below.
As we have proved above, the TE values are proportional to $S^2$. Besides, previous experimental works observed that the structural connectivity $S$ follows the log-normal distribution both for cortico-cortical and local networks in mouse and monkey brains \cite{bm2014log}. Thus, the distribution of TE should also follow the log-normal distribution. Therefore, we fit the distribution of TE values for experimental data with the summation of two log-normal likelihood functions (details in \hyperref[sec:Materials-and-Methods]{\textit{Materials and Methods}}).

\begin{figure}[!htp]
	\begin{centering}
		\includegraphics[width=1\linewidth]{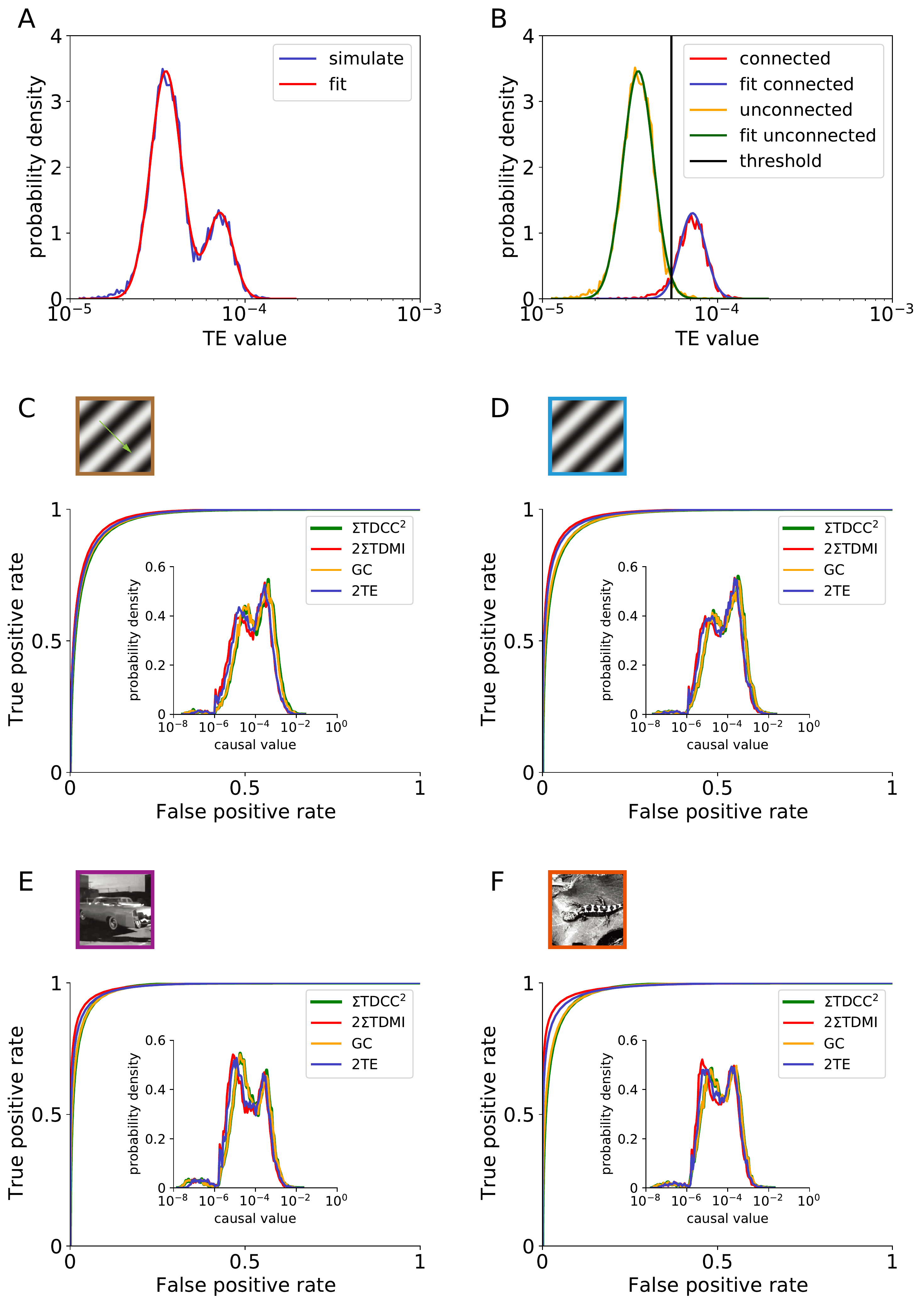}
		\par\end{centering}
	\caption{
		Reconstruction of structural connectivity by the assumption of  mixed log-normal distribution of causal values. (\textit{A}) Distribution of TE values in an HH network of 100 excitatory     neurons receiving correlated external Poisson input. The correlation coefficient of the Poisson input to each neuron is 0.3. The blue and red curves are the simulated and fitted     distributions, respectively. (\textit{B}) Distributions of TE values from connected and unconnected pairs of HH neurons in (\textit{A}). The red and orange curves are the simulated TE values from connected and unconnected pairs respectively while the blue and green curves are the fitted TE values from connected and unconnected pairs respectively which are obtained from the fitting in (\textit{A}). The parameters are the same as those in Fig. \ref{fig:distribution}. (\textit{C-F}): ROC curves of the network composed by the recorded neurons in experiment under visual stimuli of (\textit{C}) drifting gratings, (\textit{D}) static gratings, (\textit{E}) natural movie, and (\textit{F}) natural scenes. Insets:  the distribution of causal values of each pair of neurons in the whole network. We use the experimental spike data (sections id 715093703 at \url{https://allensdk.readthedocs.io/}) with signal-to-noise ratio greater than 4 and firing rate greater than 0.05 Hz. The parameters are set as  $k=1,l=10$, $\Delta t=1$ ms, and $m=1$. The colors are the same as those in Fig. \ref{fig:parameters}. 
		\label{fig:Allen}}
\end{figure}

As an illustrative example, we consider an HH network where neurons receive correlated background Poisson input. In this case, although the distribution of TE values $\log_{10}T_{Y\rightarrow X}$ (Fig. \ref{fig:Allen}\textit{A}) for synaptically connected and unconnected pairs of neurons overlap with each other, it is still well fitted by the summation of two log-normal distributions of TE values. 
More importantly, the fitted distribution of TE from connected and unconnected pairs agreed well with those from the true connectivity used in simulation (Fig. \ref{fig:Allen}\textit{B}).
Thus, we take the fitted causality measure distributions as the ground truth of structural connectivity to evaluate the performance of network connectivity reconstruction, $e.g.$, the AUC value is 0.997 in Fig. \ref{fig:Allen}\textit{B}.
In addition, we also provide the optimal inference threshold of TE values as the intersection of the fitted curves \cite{dayan2001theoretical} indicated by the vertical solid line in Fig. \ref{fig:Allen}\textit{B}.

We then apply the same ROC analysis to the causality measures from experimental data with different visual stimuli conditions. As shown in the insets of Figs. \ref{fig:Allen}\textit{C-F}, the distribution of TDCC, TDMI, GC, and TE values are very close, which again verifies their mathematical relations given by Theorems \ref{thm: MI vs CC}-\ref{thm: TE vs GC}. 
Under the assumption of mixed log-normal distribution (\textit{\textcolor{blue}{SI Appendix}}, \textcolor{blue}{Figs. S7A-D}), we first infer the ground truth of the network structural connectivity, and then evaluate the reconstruction performance of TDCC, TDMI, GC, and TE by the ROC curves which nearly overlap as shown in Figs. \ref{fig:Allen}\textit{C-F} with AUC values greater than 0.97. 
We also infer the binary adjacency matrix using the optimal inference threshold (\hyperref[sec:Materials-and-Methods]{\textit{Materials and Methods}}) from ROC analysis. The reconstructed adjacency matrix from all four different condition of visual stimuli are consistent, in which more than 89\% of the total pairs are consistently classified into true positive and true negative population of inferred connections.
Almost all of the causal values of inconsistently inferred connections fall into the overlapping region of the connected and unconnected distributions, $i.e.$, the causal values from connected (unconnected) pairs but less (greater) than the threshold(green curve in \textit{\textcolor{blue}{SI Appendix}}, \textcolor{blue}{Figs. S7E-H}) which are generally error prone. 
Comparing the reconstructed structural connectivity of an HH network driven by different correlated Poisson inputs, the inconsistent causal values also mainly fall on the overlapping region (\textit{\textcolor{blue}{SI Appendix}}, \textcolor{blue}{Fig. S8}).

\section*{Discussion }
In this work, we have revealed the quantitative relations among four widely used causality measures (TDCC, TDMI, GC, and TE) for pulse-like signals, and have demonstrated their capability for reconstructing the physical (structure) connectivity using the example of HH neural network. 
Meanwhile, these causal inference methods also can be successfully applied to reconstruct the connectivity of a subnetwork while the neuronal activities beyond the subnetwork are not observable. 
Then, we have applied them to reconstruct the structural connectivity of the real neuronal network in mouse brain using spike train data, which have been massively recorded by calcium imaging \cite{stosiek2003vivo,grewe2010high,dana2019high} or single and multiple electrodes \cite{field2010functional,jun2017fully,stringer2019spontaneous} in experiments, and have achieved promising performances. 

To emphasize the essence of effective reconstruction for pulse-like signals in our framework, we again address the fact that i) they have small correlation length, and ii) their indirect causalities are much weaker than direct ones. 
On the one hand, the auto-correlation function of pulse-like signals rapidly decay with the shrinkage of time step $\Delta t$ (\textit{\textcolor{blue}{SI Appendix}} \textcolor{blue}{Fig. S1}). It protects the inferred causality from the corruption of the self-memory of time series. Also, it allows us to use a small $k$ and $l$ ($k=l=1$) in application which overcomes the curse of dimensionality in the estimation of probability distribution and provides a practical approach for network reconstruction. 
On the other hand, the value of indirect causality is as several orders of magnitude smaller than direct one, which ensures the feasibility of distinguishing the directed connections from indirect ones according to their magnitudes of causal connectivity.
With those property of pulse-like signals, it is worth pointing out that the application of the simplest statistic, TDCC, is sufficient enough to reconstruct the connectivity of neuronal networks with spike trains as measured signals.
In contrast, if these causality measures are directly applied to continuous-valued signals, $e.g.$, voltage time series, the mathematical relations derived in our theorems become invalid (\textit{\textcolor{blue}{SI Appendix}}, \textcolor{blue}{Fig. S9A}). Moreover, TDCC and TDMI may give incorrect reconstruction of the structural connectivity due to the strong self-correlation of continuous-valued time series (\textit{\textcolor{blue}{SI Appendix}}, \textcolor{blue}{Fig. S9B}).


Although we have illustrated the effectiveness of the four causality measures taking the examples of excitatory HH neural network receiving uncorrelated external Poisson drive (\hyperref[sec:Materials-and-Methods]{\textit{Materials and Methods}}), we find that these measures are applicable to more general situations, including networks in different dynamical states, networks receiving correlated inputs, networks with both excitatory and inhibitory neurons, networks with different neuronal model, and other pulse-output networks beyond neural systems.

First, oscillations and synchronizations are commonly observed in the biological brain network, as shown in \textit{\textcolor{blue}{SI Appendix}}, \textcolor{blue}{Fig. S10A}. Due to the fake causality between neurons introduced by the strong synchronous state, conventional reconstruction frameworks fail to capture the true structural connectivity. However, with our framework, high inference accuracy (AUC $>$ 0.88) can still be achieved (shown in \textit{\textcolor{blue}{SI Appendix}}, \textcolor{blue}{Figs. S10B-C}). Furthermore, by applying a desynchronized sampling method that only samples the pulse-output signals in the asynchronous time interval (\textit{\textcolor{blue}{SI Appendix}}, \textcolor{blue}{Fig. S10D}), we can again perfectly reconstruct the network (AUC $>$ 0.99), shown in \textit{\textcolor{blue}{SI Appendix}}, \textcolor{blue}{Figs. S10E-F}. 
The relation between percentages of desynchronized downsampling and AUC values is shown in \textit{\textcolor{blue}{SI Appendix}}, \textcolor{blue}{Fig. S11}.

Second, external inputs of the network in the brain can be often correlated. In such a case, the synchronized states can be achieved similarly compared with previous cases (shown in \textit{\textcolor{blue}{SI Appendix}}, \textcolor{blue}{Fig. S12A}). Similarly, our framework can still achieve high inference accuracy (AUC $>$ 0.99 with desynchronized downsampling methods, or AUC $>$ 0.88 without downsampling) if the external inputs are moderately correlated ($e.g.$, correlation coefficient less than 0.35 in our simulation case, see \textit{\textcolor{blue}{SI Appendix}}, \textcolor{blue}{Figs. S12 B-F}). 

Third, for the more general networks containing both excitatory and inhibitory neurons, we show that performance of the reconstruction remains the same for an HH network of both excitatory and inhibitory neurons (AUC $>$ 0.99 with desynchronized downsampling methods, or AUC $>$ 0.71 without downsampling) (\textit{\textcolor{blue}{SI Appendix}}, \textcolor{blue}{Fig. S13}). 

Fourth, we also apply our framework of reconstruction onto other types of neuronal networks. Here we take the current-based leaky integrate-and-fire (LIF) neuronal network containing 100 excitatory LIF neurons (\hyperref[sec:Materials-and-Methods]{\textit{Materials and Methods}}) randomly connected with probability 0.25 as an example (\textit{\textcolor{blue}{SI Appendix}}, \textcolor{blue}{Fig. S14A}). 
Our framework still works for LIF network, with well separated two-modal distribution of causality values (\textit{\textcolor{blue}{SI Appendix}}, \textcolor{blue}{Fig. S14B}) and high reconstruction performance (AUC $>$ 0.99) (\textit{\textcolor{blue}{SI Appendix}}, \textcolor{blue}{Fig. S14C}).

Last but not least, we point out that the mathematical relations among four causality measures and their successful application to reconstruct the structural connectivity is general for a wide range of pulse-output nonlinear networks. Here, we apply our framework onto the pulse-output Lorenz networks (\hyperref[sec:Materials-and-Methods]{\textit{Materials and Methods}}) with 100 pulse-coupled Lorentz nodes, which is proposed for atmospheric convection with chaotic property \cite{lorenz1963deterministic}. The raster of pulse-output signals for the Lorentz network is shown in \textit{\textcolor{blue}{SI Appendix}}, \textcolor{blue}{Fig. S14D}. Again, a well separated two-modal distribution of causality values (\textit{\textcolor{blue}{SI Appendix}}, \textcolor{blue}{Fig. S14E}) and a high reconstruction performance (AUC $>$ 0.99)(\textit{\textcolor{blue}{SI Appendix}}, \textcolor{blue}{Fig. S14F}) prove the validity of our inference framework.

	
	\subsection*{Materials-and-Methods} 
	
	\subsection*{Reconstruction performance evaluation}
	For binary inference of structural connectivity, analysis based on receiver operating characteristic (ROC) curves is adopted to evaluate the reconstruction performance in this work. The following two scenarios are considered.
	
	\subsubsection*{With known structural connectivity}
	The conventional procedures for ROC curves analysis can be naturally applied towards data with true labels, i.e. structural connectivity in our binary reconstruction case. The area under the ROC curve (AUC) quantifies how well causality measures can distinguish from structurally connected edges from the other. If AUC is close to 1, the distributions of the causality measures between those of connected edges and disconnected ones are well distinguishable, i.e., the performance of binary reconstruction is good. If the AUC is close to 0.5, the distribution of those two kinds are virtually indistinguishable, meaning the performance of reconstruction is close to random guess.
	
	\subsubsection*{Without known structural connectivity}
	The causality measures for synaptically connected and unconnected pairs of neurons can be fitted by the log-normal distribution with different parameters, $i.e.$,  $\log_{10}T_{Y\rightarrow X}^{(\text{con})}\sim\mathcal{N}(\mu_{1},\sigma_{1}^{2})$ and $\log_{10}T_{Y\rightarrow X}^{(\text{uncon})}\sim\mathcal{N}(\mu_{2},\sigma_{2}^{2})$ where the superscripts ``con'' and ``uncon'' indicates a directed synaptic connection and no synaptic connection from $Y$ to $X$, respectively. Thus, the overall distribution is fitted with the summation of two log-normal distributions
	
	\begin{equation*} \label{density fit}
	\frac{p_{c}}{\sqrt{2\pi\sigma_{1}^{2}}}\exp\left(-\frac{(x-\mu_{1})^{2}}{2\sigma_{1}^{2}}\right)+\frac{1-p_{c}}{\sqrt{2\pi\sigma_{2}^{2}}}\exp\left(-\frac{(x-\mu_{2})^{2}}{2\sigma_{2}^{2}}\right),
	\end{equation*}
	where $p_c$ is the proportional coefficient between two Gaussian distributions $\log_{10}T_{Y\rightarrow X}^{(\text{con})}$ and $\log_{10}T_{Y\rightarrow X}^{(\text{uncon})}$.
	After that, the fitted distribution of two types of connections are regarded as the true labels of structural connectivity. The similar procedures, as previous case, are applied to get ROC curve and AUC value to evaluate the reconstruction performance of causality measures.
	
	\subsection*{The HH model}
	The dynamics of the $i$th neuron of an HH network with $N$ neurons is governed by 
	
	\begin{equation*}\label{eq: V of HH}
	\begin{aligned}
	C\frac{dV_{i}}{dt}&=-G_{L}(V_{i}-V_{L})+I^\textrm{Na}_{i}+I^\textrm{K}_{i}+I_{i}^{\textrm{input}}\\
	I_i^\textrm{Na} &= -G_\textrm{Na}m_{i}^{3}h_{i}(V_{i}-V_\textrm{Na})\\
	I_i^\textrm{K} &= -G_\textrm{K}n_{i}^{4}(V_{i}-V_\textrm{K})\\
	\end{aligned}
	\end{equation*}
	
	\begin{equation*}\label{eq:mhn of HH}
	\frac{dz_{i}}{dt}=(1-z_{i})\alpha_{z}(V_{i})-z_{i}\beta_{z}(V_{i}),\,\,\,\text{ for }z=m,h,n,
	\end{equation*}
	where \textit{$C$} and $V_{i}$ are the neuron's membrane capacitance and membrane potential, respectively; $m_{i}$, $h_{i}$, and $n_{i}$ are gating variables; $V_\textrm{Na},V_\textrm{K}$, and $V_\textrm{L}$ are the reversal potentials for the sodium, potassium, and leak currents, respectively; $G_\textrm{Na},G_\textrm{K}$,
	and $G_\textrm{L}$ are the corresponding maximum conductance; and $\alpha_{z}$ and $\beta_{z}$ are the rate variables. 
	The detailed dynamics of the gating variables $m$, $h$, $n$ can be found in Ref. \cite{dayan2001theoretical} and \textit{\textcolor{blue}{SI Appendix, Supplementary Information Text 3}}. 
	The input current $I_{i}^{\textrm{input}}=-G_{i}(t)(V_{i}-V_{E})$
	where $G_{i}(t)$ is the input conductance defined as $G_{i}(t)=f\sum_{l}H(t-s_{il})+\sum_{j}A_{ij}S\sum_{l}H(t-\tau_{jl})$
	and $V_{E}$ is the reversal potential of excitation. Here, $s_{il}$ is the $l$th spike time of the external Poisson input with strength $f$ and rate $\nu$, $\textbf{A}=(A_{ij})$ is the
	adjacency matrix with $A_{ij}=1$ indicating a directed connection
	from neuron $j$ to neuron $i$ and $A_{ij}=0$ indicating no connection there.
	$S$ is the coupling strength, and $\tau_{jl}$ is the $l$th spike
	time of the $j$th neuron. The spike-induced conductance change $H(t)$ is defined
	as 
	$
	H(t)=\frac{\sigma_{d}\sigma_{r}}{\sigma_{d}-\sigma_{r}}\left[\exp\left(-\frac{t}{\sigma_{d}}\right)-\exp\left(-\frac{t}{\sigma_{r}}\right)\right]\Theta(t),
	$
	where $\sigma_{d}$ and $\sigma_{r}$ are the decay and
	rise time scale, respectively, and $\Theta(\cdot)$ is the Heaviside
	function.
	When the voltage $V_{j}$ reaches the firing
	threshold, $V^{\textrm{th}}$, the $j$th neuron generates a spike
	at this time, say $\tau_{jl}$, and it will induce the $i$th neuron's conductance change if $A_{ij}=1$.
	
	\subsection*{Current-based LIF model}
	The dynamics of the $i$th current-based leaky integrate-and-fire (LIF) neuron is governed by
	
	\begin{equation*}\label{eq: V of LIF}
	\begin{aligned}
	C\frac{dV_{i}}{dt}&=-G_{L}(V_{i}-V_{L})+I_{i}^{\textrm{input}}\\
	V_i(t) &= V_{\textrm{reset}}\quad \textrm{if} \quad V_i(t)\geq V^\textrm{th},\\
	\end{aligned}
	\end{equation*}
	where \textit{$C$} and $V_{i}$ are the membrane capacitance and membrane potential (voltage). $V_{L}$ and $G_{L}$ are the reversal potentials and conductance for leak currents. 
	Compared with HH model, current-based LIF model drops terms of nonlinear sodium and potassium current, and replaces the conductance-based input current with the current-based one, $I_{i}^{\textrm{input}} = f\sum_l\delta(t-s_{il})+\sum_jA_{ij}S\sum_l \delta\left(t-\tau_{jl}\right)$, where $s_{il}$ is the $l$th spike time of external Poisson input with strength $f$ and rate $\nu$, and $\tau_{jl}$ is the $l$th spike time of $j$th neuron with strength $S$. And $\mathbf{A} = A_{ij}$ is the adjacency matrix defined the same as those in the HH model. 
	When the voltage reaches the threshold $V^\textrm{th}$, the $i$th neuron will emit a spike to all its connected post-synaptic neurons, and then reset to $V_\textrm{reset}$ immediately. In numerical simulation, we use dimensionless quantities: $C=1$, $V_{R}=0$, $V^{\text{th}}=1$, and the leakage conductance is set to be $G_{L}=0.05 \textrm{ms}^{-1}$ corresponding to the membrane time constant of $20\textrm{ms}$. 
	
	\subsection*{Pulse-output Lorenz network}
	The dynamics of the $i$th node of a Lorenz network is governed by
	
	\begin{equation*}
	\begin{aligned}
	\frac{dx_{i}}{dt} &= \sigma(y_{i}-x_{i}) + \sum_{j}A_{ij}S\sum_{l} \delta(t-\tau_{jl})\\
	\frac{dy_{i}}{dt} & =\rho x_{i}-y_{i}-x_{i}z_{i}\\
	\frac{dz_{i}}{dt} & =-\beta z_{i}+x_{i}y_{i},\\
	\end{aligned}
	\end{equation*}
	where $\sigma=10,\beta=8/3,\rho=28$. $\tau_{jl}$ is the $l$-th output time of the $j$-th node determined by the following. When $x_{j}$ reaches a threshold $x^{\textrm{th}}=10$, it generates a pulse that induces a change in $x$ of all of its post nodes. 
	
	\subsection*{Neurophysiological data} The public spike train data is from Allen brain observatory \cite{allen2016allen}, accessed via the Allen Software Development Kit (AllenSDK) \cite{allen2015allen}. Specifically, the data labeled with sessions-ID 715093703 was analyzed in this work. 
	The 118-day-old male mouse passively received multiple visual stimuli from one of four categories, including drift gratings, static gratings, natural scenes and natural movies. The single neuronal activities, i.e. spike trains, were recorded from 13 distinct brain areas, including APN, CA1, CA3, DG, LGd, LP, PO, VISam, VISl, VISp, VISpm, VISrl, and grey, using 6 Neuropixel probes. For each category of stimulus, the recording lasts for more than 40 minutes. 
	884 sorted spike trains were recorded in this session and 131 of those were used for causality analysis with signal-to-noise ratio greater than 4.





\begin{small}	
\bibliographystyle{unsrt}
\bibliography{MyCollection}	
\end{small}

\end{document}


\title{Quantitative relations among causality measures with
	applications to nonlinear pulse-output network reconstruction}

\author{Zhong-qi K. Tian, Kai Chen, Songting Li, David W. McLaughlin, and Douglas Zhou}

\date{}

\maketitle

\section{Mathematical derivation of relations among four causality measures}
We first derive the mathematical relations among time-delayed correlation coefficient (TDCC),
time-delayed mutual information (TDMI), Granger causality (GC), and
transfer entropy (TE)
for networks with pulse signals as measured output. Consider a pair
of nodes in a network, say nodes $X$ and $Y$ with pulse-output signals $w_{x}(t)=\sum_{l}\delta(t-\tau_{xl})$
and $w_{y}(t)=\sum_{l}\delta(t-\tau_{yl})$, where $\delta(\cdot)$
is the Dirac delta function. Due to the sampling limit,
the pulse-output signals are measured as 
binary time series $\{x_{n}\}$ and $\{y_{n}\}$ with a sampling 
resolution $\Delta t$, where $x_n=1$ ($y_n=1$) if there is a pulse signal of $X$ ($Y$) in the time window $[t_n,t_{n+1})$, and $x_n=0$ ($y_n=0$) otherwise, $i.e.$,
\begin{equation*}
x_n=\int_{t_n}^{t_{n+1}}w_x(t)dt \,\,\,\, \text{ and } \,\,\,\, y_n=\int_{t_n}^{t_{n+1}}w_y(t)dt,
\end{equation*}
and $t_n=n\Delta t$. Note that the value of $\Delta t$ is often chosen to make sure that there is at most one pulse signal in one time window. The responses $x_n$ and $y_n$ can be viewed as stochastic processes when the network is driven by stationary stochastic inputs. Accordingly, below we will describe the neuronal response using the tool of probability.

For ease of discussion, we give the notations:
\begin{equation*} 
r_{x}=\frac{1}{T}\int_{0}^{T}w_{x}(t)dt \,\,\,\, \text{ and } \,\,\,\, r_{y}=\frac{1}{T}\int_{0}^{T}w_{y}(t)dt
\end{equation*}
are the mean pulse rates of $X$ and $Y$, respectively;
\begin{equation*}
p_{x}=p(x_{n}=1)  \,\,\,\, \text{ and } \,\,\,\, p_{y}=p(y_{n}=1)
\end{equation*} 
are the probability of $x_n$ and $y_n$ being 1, respectively;  
\begin{equation*}
\Delta p_{m}=\frac{p(x_{n}=1,y_{n-m}=1)}{p(x_{n}=1)p(y_{n-m}=1)}-1
\end{equation*}
measures the dependence between $x_{n}$ and $y_{n-m}$. 
Then we have 
\begin{equation}
p_{x}=r_{x}\Delta t=O(\Delta t),\,\,\,  p_{y}=r_{y}\Delta t=O(\Delta t)
\label{eq:px py}
\end{equation}
and 
\begin{equation*}
\sigma_{x}^{2}=p_{x}-p_{x}^{2}=O(\Delta t),\,\,\,  \sigma_{y}^{2}=p_{y}-p_{y}^{2}=O(\Delta t),
\end{equation*}
where the symbol $``O"$ stands for the order, $\sigma_{x}$ and $\sigma_{y}$ are the standard
deviation of $\{x_{n}\}$ and $\{y_{n}\}$, respectively.

\subsection{Definition of TDCC, TDMI, GC, and TE}
Without loss of generality, we consider the causal interaction from $Y$ to $X$ with binary time series $\{x_{n}\}$ and $\{y_{n}\}$.
TDCC from $Y$ to $X$ is defined by   
\begin{equation}
C(X,Y;m)=\frac{\text{cov}(x_{n},y_{n-m})}{\sigma_{x}\sigma_{y}},\label{eq:TDCC}
\end{equation}
where $m$ is time delay.

TDMI from $Y$ to $X$ is defined by 
\begin{equation*}
I(X,Y;m)=\sum_{x_{n},y_{n-m}}p(x_{n},y_{n-m})\log\frac{p(x_{n},y_{n-m})}{p(x_{n})p(y_{n-m})},\label{eq:MI}
\end{equation*}
where $p(x_{n},y_{n-m})$ is the joint probability distribution of
$x_{n}$ and $y_{n-m}$, $p(x_{n})$ and $p(y_{n-m})$ are the corresponding
marginal probability distributions.

GC is established based on linear regression. The auto-regression for $X$ is
represented by 
\[x_{n+1}=a_{0}+\sum_{i=1}^{k}a_{i}x_{n+1-i}+\epsilon_{n+1},\]
where $\{a_{i}\}$ are the auto-regression coefficients and $\epsilon_{n+1}$
is the residual. By including the historical information of $Y$ with message
length $l$ and time delay $m$, the joint regression for $X$ is
represented by 
\[x_{n+1}=\tilde{a}_{0}+\sum_{i=1}^{k}\tilde{a}_{i}x_{n+1-i}+\sum_{j=1}^{l}b_{j}y_{n+2-m-j}+\eta_{n+1},\]
where $\{\tilde{a}_{i}\}$ and $\{b_{j}\}$ are the joint regression
coefficients, and $\eta_{n+1}$ is the corresponding residual. 
The GC value from $Y$ to $X$ is defined by 
\begin{equation*}
G_{Y\rightarrow X}(k,l;m)=\log\frac{\textrm{Var}(\epsilon_{n+1})}{\textrm{Var}(\eta_{n+1})}.\label{eq:GC}
\end{equation*}
By introducing the time-delay parameter $m$, the GC analysis defined above generalizes the conventional GC analysis, as the latter corresponds to the special case of $m=1$. 

The TE value from $Y$ to $X$ is defined by
\begin{equation}
T_{Y\rightarrow X}(k,l;m)  =\sum_{x_{n+1},x_{n}^{(k)},y_{n+1-m}^{(l)}}p(x_{n+1},x_{n}^{(k)},y_{n+1-m}^{(l)})\log\frac{p(x_{n+1}|x_{n}^{(k)},y_{n+1-m}^{(l)})}{p(x_{n+1}|x_{n}^{(k)})},
\label{eq:TE}\end{equation}
where the shorthand notation $x_{n}^{(k)}=(x_{n},x_{n-1},...,x_{n-k+1})$
and $y_{n+1-m}^{(l)}=(y_{n+1-m},y_{n-m},...,y_{n+2-m-l})$, $k,l$
indicate the length (order) of historical information of $X$ and $Y$,
respectively. Similar to GC, the time-delay parameter $m$ is introduced that generalizes the conventional TE, the latter of which corresponds to the case of $m=1$.

\subsection{Mathematical relation between TDMI and TDCC}
From the definition of TDCC in Eq. \ref{eq:TDCC}, we have
\begin{equation} \label{eq:CC dp}
\begin{aligned}C(X,Y;m) & =\frac{\text{cov}(x_{n},y_{n-m})}{\sigma_{x}\sigma_{y}}\\
& =\frac{p(x_{n}=1,y_{n-m}=1)-p_{x}p_{y}}{\sqrt{(p_{x}-p_{x}^{2})(p_{y}-p_{y}^{2})}}.
\end{aligned}
\end{equation}
The relation between TDMI and TDCC can be derived by Taylor expanding TDMI with respect to the difference $\frac{p(x_{n},y_{n-m})}{p(x_{n})p(y_{n-m})}-1$ as follows:

\begin{equation}
\begin{alignedat}{1}I(X,Y;m) =& \sum_{x_{n},y_{n-m}}p(x_{n})p(y_{n-m})\left[1+\left(\frac{p(x_{n},y_{n-m})}{p(x_{n})p(y_{n-m})}-1\right)\right]\log\left[1+\left(\frac{p(x_{n},y_{n-m})}{p(x_{n})p(y_{n-m})}-1\right)\right]\\
=& \sum_{x_{n},y_{n-m}}p(x_{n})p(y_{n-m})\left[\left(\frac{p(x_{n},y_{n-m})}{p(x_{n})p(y_{n-m})}-1\right)+\frac{1}{2}\left(\frac{p(x_{n},y_{n-m})}{p(x_{n})p(y_{n-m})}-1\right)^{2}\right.\\
& +\left.O\Biggl(\left(\frac{p(x_{n},y_{n-m})}{p(x_{n})p(y_{n-m})}-1\right)^{3}\Biggr)\right]\\
=& \frac{[p(x_{n}=1,y_{n-m}=1)-p_{x}p_{y}]^{2}}{2(p_{x}-p_{x}^{2})(p_{y}-p_{y}^{2})}+O(\Delta t^{2}\Delta p_{m}^{3})\\
=& \frac{C^{2}(X,Y;m)}{2}+O(\Delta t^{2}\Delta p_{m}^{3}).
\end{alignedat}
\label{eq:MI vs CC}
\end{equation}

\subsection{Mathematical relation between GC and TDCC}
From the definition, GC can be represented by the covariances of the signals \cite{Barnett2009}
as 
\begin{equation}
\begin{aligned}G_{Y\rightarrow X}(k,l;m) & =\log\frac{\Gamma(x_{n+1}|x_{n}^{(k)})}{\Gamma(x_{n+1}|x_{n}^{(k)}\oplus y_{n+1-m}^{(l)})},\end{aligned}
\label{eq:GC comp}
\end{equation}
where $\Gamma(\mathbf{x}|\mathbf{y})=\textrm{cov}(\mathbf{x})-\textrm{cov}(\mathbf{x},\mathbf{y})\textrm{cov}(\mathbf{y})^{-1}\textrm{cov}(\mathbf{x},\mathbf{y})^{T}$
for random vectors $\mathbf{x}$ and $\mathbf{y}$, $\textrm{cov}(\mathbf{x})$ and $\textrm{cov}(\mathbf{y})$
denote the covariance matrix of $\mathbf{x}$ and $\mathbf{y}$, respectively, and $\textrm{cov}(\mathbf{x},\mathbf{y})$
denotes the cross-covariance matrix between $\mathbf{x}$ and $\mathbf{y}$. The symbol $T$ is the transpose
operator and $\oplus$ denotes the concatenation of vectors. 

To derive the relation between GC and TDCC, we first consider the auto-correlation
function (ACF) of the binary time series $\{x_{n}\}$ and $\{y_n\}$. The ACF of  $\{x_{n}\}$ is defined as

\begin{equation*}
\textrm{ACF}(X;m)=\frac{\text{cov}(x_{n},x_{n-m})}{\sigma_{x}^{2}},
\end{equation*}
where $m$ is non-zero time delay. Let $g_{x}(t)$ be the probability
density function that node $X$ will generate a pulse at time $t$ given
that it produced a pulse at time $t=0$. Then we have 
\[p(x_{n}=1|x_{n-m}=1)=g_{x}(m\Delta t)\Delta t+O(\Delta t^2).\]
In general, the function $g_{x}(t)$ is continuous and bounded, thus
we have $p(x_n=1|x_{n-m}=1)=O(\Delta t)$. Together with Eq. \ref{eq:px py}, we 
can obtain
\begin{equation}
\textrm{ACF}(X;m)=\frac{p(x_{n}=1|x_{n-m}=1)-p_{x}}{1-p_{x}}=O(\Delta t).\label{eq:ACF vs dt}
\end{equation}
Similarly, we have
\begin{equation*}
\textrm{ACF}(Y;m)=O(\Delta t).\label{eq:ACF vs dt1}
\end{equation*}
Based on this, we derive the relation between GC and TDCC as follows:
from Eq. \ref{eq:ACF vs dt}, we can obtain 

\begin{equation*}
\begin{aligned}\textrm{cov}(x_{n}^{(k)}) & =\sigma_{x}^{2}(\mathbf{I}+\hat{\mathbf{A}}),\end{aligned}
\end{equation*}

\begin{equation*}
\textrm{cov}(x_{n}^{(k)})^{-1}=\frac{1}{\sigma_{x}^{2}}(\mathbf{I}-\hat{\mathbf{A}})+O(\Delta t\mathbf{1}_{k\times k}),
\end{equation*}
where $\hat{\mathbf{A}}=(\hat{a}_{ij}),\hat{a}_{ij}=O(\Delta t)$, $\mathbf{I}$
is the identity matrix, and $\mathbf{1}_{k\times k}$ is the all-one
matrix. Thus,

\begin{equation}
\begin{aligned}\Gamma(x_{n+1}|x_{n}^{(k)}) & =\sigma_{x}^{2}-\frac{1}{\sigma_{x}^{2}}\textrm{cov}(x_{n+1},x_{n}^{(k)})(\mathbf{I}-\hat{\mathbf{A}})\textrm{cov}(x_{n+1},x_{n}^{(k)})^{T}+O(\Delta t^{5}).\end{aligned}
\label{eq:derivation GC CC 1}
\end{equation}
In the same way, we have

\begin{equation*}
\text{cov}(x_{n}^{(k)}\oplus y_{n+1-m}^{(l)})=\left(\begin{array}{cc}
\sigma_{x}^{2}(\mathbf{I}+\hat{\mathbf{A}}) & \sigma_{x}\sigma_{y}\hat{\mathbf{C}}\\
\sigma_{x}\sigma_{y}\hat{\mathbf{C}}^{T} & \sigma_{y}^{2}(\mathbf{I}+\hat{\mathbf{B}})
\end{array}\right),
\end{equation*}

\begin{equation*}
\text{cov}(x_{n}^{(k)}\oplus y_{n+1-m}^{(l)})^{-1}=\left(\begin{array}{cc}
(\mathbf{I}-\hat{\mathbf{A}})/\sigma_{x}^{2} & -\hat{\mathbf{C}}/\sigma_{x}\sigma_{y}\\
-\hat{\mathbf{C}}^{T}/\sigma_{x}\sigma_{y} & (\mathbf{I}-\hat{\mathbf{B}})/\sigma_{y}^{2}
\end{array}\right)+O\left(\Delta t\mathbf{1}_{(k+l)\times(k+l)}\right),
\end{equation*}
where $\hat{\mathbf{B}}=(\hat{b}_{ij}),\hat{b}_{ij}=O(\Delta t), \hat{\mathbf{C}}=(\hat{c}_{ij}),\hat{c}_{ij}=O(\Delta t\Delta p_{m}).$
Thus,

\begin{equation}
\begin{aligned}\Gamma(x_{n+1}|x_{n}^{(k)}\oplus y_{n+1-m}^{(l)}) =&\sigma_{x}^{2}-\frac{1}{\sigma_{x}^{2}}\textrm{cov}(x_{n+1},x_{n}^{(k)})(\mathbf{I}-\hat{\mathbf{A}})\textrm{cov}(x_{n+1},x_{n}^{(k)})^{T}\\
& -\frac{1}{\sigma_{y}^{2}}\textrm{cov}(x_{n+1},y_{n+1-m}^{(l)})(\mathbf{I}-\hat{\mathbf{B}})\textrm{cov}(x_{n+1},y_{n+1-m}^{(l)})^{T}+O(\Delta t^{4}\Delta p_{m}^{2}).
\end{aligned}
\label{eq:derivation GC CC 2}
\end{equation}
Substituting Eqs. \ref{eq:derivation GC CC 1} and \ref{eq:derivation GC CC 2} into Eq. \ref{eq:GC comp} and Taylor expanding Eq. \ref{eq:GC comp} with respect to $\Delta t$, we can obtain                     
\begin{equation} \label{eq:GC vs CC}
\begin{aligned}G_{Y\rightarrow X}(k,l;m) & =\frac{\textrm{cov}(x_{n+1},y_{n+1-m}^{(l)})\textrm{cov}(x_{n+1},y_{n+1-m}^{(l)})^{T}}{\sigma_{x}^{2}\sigma_{y}^{2}}+O(\Delta t^{3}\Delta p_{m}^{2})\\
& =\sum_{i=m}^{m+l-1}C^{2}(X,Y;i)+O(\Delta t^{3}\Delta p_{m}^{2}).
\end{aligned}
\end{equation}
Note that we omit the higher order term $O(\Delta t^4)$ in the above derivation.

\subsection{Mathematical relation between TE and TDMI}
To rigorously establish the relation between TE and TDMI,
we require  that $\Bigl\Vert x_{n+1}^{(k+1)}\Bigr\Vert_{0}\leq1$ and $\Bigl\Vert y_{n+1-m}^{(l)}\Bigr\Vert_{0}\leq1$ in the definition of TE given in Eq. \ref{eq:TE},
where $\Bigl\Vert \cdot \Bigr\Vert_{0}$ denotes the $l_{0}$ norm of a vector, $i.e.$, the number
of nonzero elements in a vector. 
This assumption indicates that the length of historical
information used in the TE framework is shorter than the ``refractory period'', $i.e.$, the minimal  time interval between
two consecutive pulse-output  signals.

To simplify the notations, we use
$x^{-}$and $y^{-}$ to denote $x_{n}^{(k)}=(x_{n},x_{n-1},...,x_{n-k+1})$
and $y_{n+1-m}^{(l)}=(y_{n+1-m},y_{n-m},...,y_{n+2-m-l})$, respectively.
From the definition of TE, we have 

\begin{equation*}
\begin{aligned}T_{Y\rightarrow X}(k,l;m) & =\sum_{x_{n+1},x^{-},y^{-}}p(x_{n+1},x^{-},y^{-})\log\frac{p(x_{n+1}|x^{-},y^{-})}{p(x_{n+1}|x^{-})}\\
& =\sum_{x_{n+1},x^{-},y^{-}}p(x_{n+1},x^{-},y^{-})\left[\log\frac{p(x_{n+1}|y^{-})}{p(x_{n+1})}+\log\frac{p(x_{n+1}|x^{-},y^{-})}{p(x_{n+1}|y^{-})}\frac{p(x_{n+1})}{p(x_{n+1}|x^{-})}\right].
\end{aligned}
\end{equation*}
Because
\begin{equation*}
\begin{aligned}\sum_{x_{n+1},y^{-}}p(x_{n+1},y^{-})\log\frac{p(x_{n+1}|y^{-})}{p(x_{n+1})} & =\sum_{x_{n+1},y^{-}}p(x_{n+1},y^{-})\log\frac{p(y^{-}|x_{n+1})}{p(y^{-})}\\
& =\sum_{x_{n+1},y^{-}}p(x_{n+1},y^{-})\left[\log\frac{\prod_{j}p(y_{j}|x_{n+1})}{\prod_{j}p(y_{j})}+\log\frac{p(y^{-}|x_{n+1})}{\prod_{j}p(y_{j}|x_{n+1})}\frac{\prod_{j}p(y_{j})}{p(y^{-})}\right]\\
& =\sum_{i=m}^{m+l-1}I(X,Y;i)+\sum_{x_{n+1},y^{-}}p(x_{n+1},y^{-})\log\frac{p(y^{-}|x_{n+1})}{\prod_{j}p(y_{j}|x_{n+1})}\frac{\prod_{j}p(y_{j})}{p(y^{-})},
\end{aligned}
\end{equation*}
where $\prod_{j}$ represents $\prod_{j=n+2-m-l}^{n+1-m}$, we have


\begin{equation}
T_{Y\rightarrow X}(k,l;m)=\sum_{i=m}^{m+l-1}I(X,Y;i)+\mathcal{A}+\mathcal{B}, 
\end{equation}
where
\[
\mathcal{A}=\sum_{x_{n+1},y^{-}}p(x_{n+1},y^{-})\log\frac{p(y^{-}|x_{n+1})}{\prod_{j}p(y_{j}|x_{n+1})}\frac{\prod_{j}p(y_{j})}{p(y^{-})}
\]
and 
\[
\mathcal{B}=\sum_{x_{n+1},x^{-},y^{-}}p(x_{n+1},x^{-},y^{-})\log\frac{p(x_{n+1}|x^{-},y^{-})}{p(x_{n+1}|y^{-})}\frac{p(x_{n+1})}{p(x_{n+1}|x^{-})}.
\]

Under the assumption that $\Bigl\Vert x_{n+1}^{(k+1)}\Bigr\Vert_{0}\leq1$ and $\Bigl\Vert y_{n+1-m}^{(l)}\Bigr\Vert_{0}\leq1$,
the number of nonzero components is at most one in $x_{n+1}^{(k+1)}$
and $y_{n+1-m}^{(l)}$.
We use $1_{x_{s}}$ to denote the event
that only the state $x_{s}$ is one in $x^{-}$, $n-k+1\leq s\leq n$
and use $0_{x}^{-}$ to denote the event that all the components
in $x^{-}$ are zero. Similarly, we use $1_{y_{t}}$ to denote
the event that only the state $y_{t}$ is one in $y^{-}$, $n+2-m-l\leq t\leq n+1-m$
and use $0_{y}^{-}$ to denote the event that all the components
in $y^{-}$ are zero. Then we can derive the  leading order of  each term in $\mathcal{A}$
and $\mathcal{B}$ by Taylor expanding them with respect to $\Delta t$ and $\Delta p_m$.
For example, in the case of 
\begin{equation} \label{eq: example in A in TE}
p(x_{n+1},y^{-})\log\frac{p(y^{-}|x_{n+1})}{\prod_{j}p(y_{j}|x_{n+1})}\frac{\prod_{j}p(y_{j})}{p(y^{-})}\biggl|_{x_{n+1}=1,y^{-}=1_{y_{t}}} 
\end{equation}
in $\mathcal{A}$, we have 

\begin{equation} \label{eq:TE derivation 1}
p(x_{n+1}=1,y^{-}=1_{y_{t}})=p(x_{n+1}=1,y_{t}=1)=p_{x}p_{y}(1+\Delta p_{n+1-t}),
\end{equation}
\begin{equation} \label{eq:TE derivation 2}
p(x_{n+1}=1,y_{j}=0)=p_{x}-p(x_{n+1}=1,y_{j}=1)=p_{x}-p_{x}p_{y}(1+\Delta p_{n+1-j}),
\end{equation}
\begin{equation} \label{eq:TE derivation 3}
p(y_{j}=0|x_{n+1}=1)=1-p_{y}(1+\Delta p_{n+1-j}).
\end{equation}
Substituting Eqs. \ref{eq:TE derivation 1}-\ref{eq:TE derivation 3} into Eq. \ref{eq: example in A in TE} yields

\[
\begin{aligned} & p(x_{n+1},y^{-})\log\frac{p(y^{-}|x_{n+1})}{\prod_{j}p(y_{j}|x_{n+1})}\frac{\prod_{j}p(y_{j})}{p(y^{-})}\biggl|_{x_{n+1}=1,y^{-}=1_{y_{t}}}\\
& =p_{x}p_{y}(1+\Delta p_{n+1-t})\log\frac{(1-p_{y})^{l-1}}{\prod_{j\neq t}(1-p_{y}-p_{y}\Delta p_{n+1-j})}\\
& =p_{x}p_{y}(1+\Delta p_{n+1-t})\sum_{j\neq t}\log\frac{1-p_{y}}{1-p_{y}-p_{y}\Delta p_{n+1-j}}\\
& =p_{x}p_{y}^{2}\sum_{j\neq t}\Delta p_{n+1-j}+O(\Delta t^{3}\Delta p_{m}^{2}),
\end{aligned}
\]
Note that we take $\Delta p_{n+1-j}=O(\Delta p_{m})$ in the above derivation.
Other terms in $\mathcal{A}$ can be obtained similarly as follows:

\begin{equation*}
\begin{aligned} & p(x_{n+1},y^{-})\log\frac{p(y^{-}|x_{n+1})}{\prod_{j}p(y_{j}|x_{n+1})}\frac{\prod_{j}p(y_{j})}{p(y^{-})}\biggl|_{x_{n+1}=1,y^{-}=0_{y}^{-}}\\
& =\left(p_{x}-p_{x}p_{y}(l+\sum_{t}\Delta p_{n+1-t})\right)\log\frac{1-p_{y}(l+\sum_{t}\Delta p_{n+1-t})}{\prod_{t}(1-p_{y}-p_{y}\Delta p_{n+1-t})}\frac{(1-p_{y})^{l}}{1-lp_{y}}\\
& =(1-l)p_{x}p_{y}^{2}\sum_{t}\Delta p_{n+1-t}+O(\Delta t^{3}\Delta p_{m}^{2}),
\end{aligned}
\end{equation*}

\begin{equation*}
\begin{aligned} & p(x_{n+1},y^{-})\log\frac{p(y^{-}|x_{n+1})}{\prod_{j}p(y_{j}|x_{n+1})}\frac{\prod_{j}p(y_{j})}{p(y^{-})}\biggl|_{x_{n+1}=0,y^{-}=1_{y_{t}}}\\
& =\left(p_{y}-p_{x}p_{y}(1+\Delta p_{n+1-t})\right)\log\frac{(1-p_{y})^{l-1}(1-p_{x})^{l-1}}{\prod_{j\neq t}\left(1-p_{x}-p_{y}+p_{x}p_{y}(1+\Delta p_{n+1-j})\right)}\\
& =-p_{x}p_{y}^{2}\sum_{j\neq t}\Delta p_{n+1-j}+O(\Delta t^{3}\Delta p_{m}^{2}),
\end{aligned}
\end{equation*}

\begin{equation*}
\begin{aligned} & p(x_{n+1},y^{-})\log\frac{p(y^{-}|x_{n+1})}{\prod_{j}p(y_{j}|x_{n+1})}\frac{\prod_{j}p(y_{j})}{p(y^{-})}\biggl|_{x_{n+1}=0,y^{-}=0_{y}^{-}}\\
& =\left(1-p_{x}-lp_{y}+p_{x}p_{y}(l+\sum_{t}\Delta p_{n+1-t})\right)\log\frac{1-p_{x}-lp_{y}+p_{x}p_{y}(l+\sum_{t}\Delta p_{n+1-t})}{\prod_{t}\left(1-p_{x}-p_{y}+p_{x}p_{y}(1+\Delta p_{n+1-t})\right)}\frac{(1-p_{y})^{l}(1-p_{x})^{l}}{(1-lp_{y})(1-p_{x})}\\
& =(l-1)p_{x}p_{y}^{2}\sum_{t}\Delta p_{n+1-t}+O(\Delta t^{3}\Delta p_{m}^{2}).
\end{aligned}
\end{equation*}
Therefore, $\mathcal{A}=O(\Delta t^{3}\Delta p_{m}^{2}).$

For $\mathcal{B},$ we have 

\begin{equation*}
\begin{aligned} & \sum_{x_{n+1},x^{-}}p(x_{n+1},x^{-})\log\frac{p(x_{n+1})}{p(x_{n+1}|x^{-})}\\
& =p_{x}\log(1-kp_{x})+kp_{x}\log(1-p_{x})+(1-(k+1)p_{x})\log\frac{(1-p_{x})(1-kp_{x})}{1-(k+1)p_{x}}\\
& =-kp_{x}^{2}-\frac{k(k+1)}{2}p_{x}^{3}+O(\Delta t^{4}),
\end{aligned}
\end{equation*}

\begin{equation*}
\begin{aligned} & p(x_{n+1},x^{-},y^{-})\log\frac{p(x_{n+1}|x^{-},y^{-})}{p(x_{n+1}|y^{-})}\biggl|_{x_{n+1}=1,x^{-}=0_{x}^{-},y^{-}=1_{y_{t}}}\\
& =-p_{x}p_{y}(1+\Delta p_{n+1-t})\log\left(1-p_{x}(k+\sum_{s}\Delta p_{s-t})\right)\\
& =p_{x}^{2}p_{y}(k+k\Delta p_{n+1-t}+\sum_{s}\Delta p_{s-t})+O(\Delta t^{3}\Delta p_{m}^{2}),
\end{aligned}
\end{equation*}

\begin{equation*}
\begin{aligned} & p(x_{n+1},x^{-},y^{-})\log\frac{p(x_{n+1}|x^{-},y^{-})}{p(x_{n+1}|y^{-})}\biggl|_{x_{n+1}=1,x^{-}=0_{x}^{-},y^{-}=0_{y}^{-}}\\
& =\left(p_{x}-p_{x}p_{y}(l+\sum_{t}\Delta p_{n+1-t})\right)\log\frac{(1-lp_{y})}{1-kp_{x}-lp_{y}+p_{x}p_{y}(kl+\sum_{s,t}\Delta p_{s-t})}\\
& =kp_{x}^{2}+\frac{k^{2}p_{x}^{3}}{2}-p_{x}^{2}p_{y}(kl+\sum_{s,t}\Delta p_{s-t}+k\sum_{t}\Delta p_{n+1-t})+O(\Delta t^{4}),
\end{aligned}
\end{equation*}

\begin{equation*}
\begin{aligned} & p(x_{n+1},x^{-},y^{-})\log\frac{p(x_{n+1}|x^{-},y^{-})}{p(x_{n+1}|y^{-})}\biggl|_{x_{n+1}=0,x^{-}=1_{x_{s}},y^{-}=1_{y_{t}}}\\
& =-p_{x}p_{y}(1+\Delta p_{s-t})\log\left(1-p_{x}(1+\Delta p_{n+1-t})\right)\\
& =p_{x}^{2}p_{y}(1+\Delta p_{s-t}+\Delta p_{n+1-t})+O(\Delta t^{3}\Delta p_{m}^{2}),
\end{aligned}
\end{equation*}

\begin{equation*}
\begin{aligned} & p(x_{n+1},x^{-},y^{-})\log\frac{p(x_{n+1}|x^{-},y^{-})}{p(x_{n+1}|y^{-})}\biggl|_{x_{n+1}=0,x^{-}=1_{x_{s}},y^{-}=0_{y}^{-}}\\
& =\left(p_{x}-p_{x}p_{y}(l+\sum_{t}\Delta p_{s-t})\right)\log\frac{1-lp_{y}}{1-p_{x}-lp_{y}+p_{x}p_{y}(l+\sum_{t}\Delta p_{n+1-t})}\\
& =p_{x}^{2}+\frac{p_{x}^{3}}{2}-p_{x}^{2}p_{y}(l+\sum_{t}\Delta p_{n+1-t}+\sum_{t}\Delta p_{s-t})+O(\Delta t^{4}),
\end{aligned}
\end{equation*}

\begin{equation*}
\begin{aligned} & p(x_{n+1},x^{-},y^{-})\log\frac{p(x_{n+1}|x^{-},y^{-})}{p(x_{n+1}|y^{-})}\biggl|_{x_{n+1}=0,x^{-}=0_{x}^{-},y^{-}=1_{y_{t}}}\\
& =\left(p_{y}-p_{x}p_{y}(k+1+\Delta p_{n+1-t}+\sum_{s}\Delta p_{s-t})\right)\log\frac{1-p_{x}(k+1+\Delta p_{n+1-t}+\sum_{s}\Delta p_{s-t})}{\left(1-p_{x}(k+\sum_{s}\Delta p_{s-t})\right)\left(1-p_{x}(1+\Delta p_{n+1-t})\right)}\\
& =-p_{x}^{2}p_{y}(k+\sum_{s}\Delta p_{s-t}+k\Delta p_{n+1-t})+O(\Delta t^{3}\Delta p_{m}^{2}),
\end{aligned}
\end{equation*}

\begin{equation*}
\begin{aligned} & p(x_{n+1},x^{-},y^{-})\log\frac{p(x_{n+1}|x^{-},y^{-})}{p(x_{n+1}|y^{-})}\biggl|_{x_{n+1}=0,x^{-}=0_{x}^{-},y^{-}=0_{y}^{-}}\\
& =\left(1-(k+1)p_{x}-lp_{y}+p_{x}p_{y}(kl+l+\sum_{s,t}\Delta p_{s-t}+\sum_{t}\Delta p_{n+1-t})\right)\\
& \cdot\log\left[\frac{1-(k+1)p_{x}-lp_{y}+p_{x}p_{y}(kl+l+\sum_{s,t}\Delta p_{s-t}+\sum_{t}\Delta p_{n+1-t})}{1-kp_{x}-lp_{y}+p_{x}p_{y}(kl+\sum_{s,t}\Delta p_{s-t})}\right.\\
& \left.\cdot\frac{1-lp_{y}}{1-p_{x}-lp_{y}+p_{x}p_{y}(l+\sum_{t}\Delta p_{n+1-t})}\right]\\
& =-kp_{x}^{2}+p_{x}^{2}p_{y}(kl+k\sum_{t}\Delta p_{n+1-t}+\sum_{s,t}\Delta p_{s-t})+O(\Delta t^{4}).
\end{aligned}
\end{equation*}
Therefore, $\mathcal{B}=O(\Delta t^{3}\Delta p_{m}^{2})$, and
thus we can obtain
\begin{equation} \label{eq:TE vs MI}
T_{Y\rightarrow X}(k,l;m)=\sum_{i=m}^{m+l-1}I(X,Y;i)+O(\Delta t^{3}\Delta p_{m}^{2}).
\end{equation}
Note that we omit the higher order term $O(\Delta t^4)$ in the above derivation.

\subsection{Mathematical relation between GC and TE}
From Eqs. \ref{eq:MI vs CC}, \ref{eq:GC vs CC}, and \ref{eq:TE vs MI}, we can straightforwardly prove the relation between GC and TE
\begin{equation*}
G_{Y\rightarrow X}(k,l;m)=2T_{Y\rightarrow X}(k,l;m)+O(\Delta t^{2}\Delta p_{m}^{3}),\label{eq:TE vs GC}
\end{equation*}
where $T_{Y\rightarrow X}$ is defined in Eq. \ref{eq:TE} with the assumption that $\Bigl\Vert x_{n+1}^{(k+1)}\Bigr\Vert_{0}\leq1$ and $\Bigl\Vert y_{n+1-m}^{(l)}\Bigr\Vert_{0}\leq1$. Note that we omit the higher order term $O(\Delta t^3)$ in the above derivation.

\section{Mechanism underlying successful network reconstruction using pairwise causal inference }
Here we demonstrate the validity of pairwise inference on pulse-output  signals in the reconstruction of network structural connectivity. It has been noticed that pairwise causal inference may potentially fail to distinguish the direct interactions from the indirect ones in a network. For example, in the case that $Y\rightarrow W\rightarrow X$ where ``$\rightarrow$'' denotes a directed connection,
the indirect interaction from Y to X may possibly be  mis-inferred as a direct interaction
via pairwise causality measures especially when the activity signals are continuous-valued. However, this type of mistake does not happen in our case of pulse-output  signals as explained
below. Here we take TDCC as an example to explain the underlying reason of successful reconstruction.

Denote \[\delta{p_{Y\rightarrow X}}=p(x_{n}=1|y_{n-m}=1)-p(x_{n}=1|y_{n-m}=0)\] as the increment of probability  of generating a pulse output by node $X$ at time step $n$ induced by a pulse-output  signal of node $Y$ at an earlier time step $n-m$.
From Eq. \ref{eq:CC dp}, we have 
\[C(X,Y;m)=\delta p_{Y\rightarrow X}\sqrt{\frac{p_{y}-p_{y}^{2}}{p_{x}-p_{x}^{2}}}.\]

Denote $S_1$ and $S_2$ as the coupling
strength from node $Y$ to node $W$ and from node $W$ to node $X$,
respectively. 
Then the increment $\delta p_{Y\rightarrow X}$ is a function of $S_{1}$
and $S_{2}$ and the Taylor expansion of $\delta p_{Y\rightarrow X}$
with respect to $S_{1}$ and $S_{2}$ has the following form
\begin{equation} \label{eq:Taylor dp}
\delta p_{Y\rightarrow X}=\alpha_{0}+\alpha_{1}S_{1}+\alpha_{2}S_{2}+\alpha_{3}S_{1}^{2}+\alpha_{4}S_{1}S_{2}+\alpha_{5}S_{2}^{2}+o(S_{1}S_{2}),
\end{equation}
where the symbol ``o'' stands for higher order terms. If $S_{1}=0$
or $S_{2}=0$, then the nodes $X$ and $Y$ are independent from the
connection structure, $i.e.,$ 
\[
\delta p_{Y\rightarrow X}\Big|_{S_{1}=0}=0 \,\,\,\, \text{ and } \,\,\,\, \delta p_{Y\rightarrow X}\Big|_{S_{2}=0}=0.
\]
Therefore, we have $\alpha_{0}=\alpha_{1}=\alpha_{2}=\alpha_{3}=\alpha_{5}=0$ in Eq. \ref{eq:Taylor dp}
and $\delta p_{Y\rightarrow X}=\alpha_{4}S_{1}S_{2}+o(S_{1}S_{2})$.
Similarly, the Taylor expansion of $\delta p_{Y\rightarrow W}$ and
$\delta p_{W\rightarrow X}$ with respect to $S_{1}$ and $S_{2}$
have the form 
\[
\delta p_{Y\rightarrow W}=\beta_{1}S_{1}+O(S_{1}^{2}) \,\,\,\, \text{ and } \,\,\,\, \delta p_{W\rightarrow X}=\beta_{2}S_{2}+O(S_{2}^{2}),
\]
respectively. Thus, we have
\begin{equation*}
\delta p_{Y\rightarrow X}=O(\delta p_{Y\rightarrow W}\cdot\delta p_{W\rightarrow X}).
\end{equation*}
Because the influence of a single input pulse signal is often small ($e.g.$, in the HH neural network with physiologically realistic coupling strengths with corresponding excitatory postsynaptic potential in the range of $[0,1]$ mV, we measure from simulation $|\delta{p}|<0.01$), the causal value $C(X,Y;m)$ from indirect interaction will be significantly smaller than $C(W,Y;m)$ or $C(X,W;m)$ from the direct interactions. Therefore, the direct connection and the indirect connection can be distinguished when performing pairwise inference on pulse-output  signals. 

\section{Detailed HH model}
\subsection{Hodgkin-Huxley (HH) neural network model of only excitatory population}

The dynamics of the $i$th neuron of an HH network with $N$ excitatory
neurons is governed by 
\begin{equation} \label{eq:HH V}
C\frac{dV_{i}}{dt}=-G_{Na}m_{i}^{3}h_{i}(V_{i}-V_{Na})-G_{K}n_{i}^{4}(V_{i}-V_{K})-G_{L}(V_{i}-V_{L})+I_{i}^{\textrm{input}},
\end{equation}

\begin{equation} \label{eq:HH mhn}
\frac{dz_{i}}{dt}=(1-z_{i})\alpha_{z}(V_{i})-z_{i}\beta_{z}(V_{i}),\,\,\,\text{ for }z=m,h,n,
\end{equation}
where $C$ is the cell membrane capacitance; $V_{i}$ is the membrane
potential (voltage); $m_{i}$, $h_{i}$, and $n_{i}$ are gating variables;
$V_{Na},V_{K}$, and $V_{L}$ are the reversal potentials for the
sodium, potassium, and leak currents, respectively; and $G_{Na},G_{K}$,
and $G_{L}$ are the corresponding maximum conductances. The rate
variables $\alpha_{z}$ and $\beta_{z}$ are defined as \cite{dayan2001theoretical}


\[
\begin{aligned}\alpha_{m}(V) & =\frac{0.1V+4}{1-\exp(-0.1V-4)}, &  & \beta_{m}(V)=4\exp\left(\frac{-(V+65)}{18}\right),\\
\alpha_{h}(V) & =0.07\exp\left(\frac{-(V+65)}{20}\right), &  & \beta_{h}(V)=\frac{1}{1+\exp(-3.5-0.1V)},\\
\alpha_{n}(V) & =\frac{0.01V+0.55}{1-\exp(-0.1V-5.5)}, &  & \beta_{n}(V)=0.125\exp\left(\frac{-(V+65)}{80}\right).
\end{aligned}
\]

The input current $I_{i}^{\textrm{input}}=-G_{i}(t)(V_{i}-V_{E})$,
where the conductance is defined as 
\begin{equation*}
G_{i}(t)=f\sum_{l}H(t-s_{il})+\sum_{j}A_{ij}S\sum_{l}H(t-\tau_{jl}),\label{eq:G}
\end{equation*}
with $V_{E}$ being the excitatory reversal potential and $s_{il}$ being
the $l$th spike time of the external Poisson input with strength $f$
and rate $\nu$. The spike-induced conductance change $H(t)$ is defined
by \cite{dayan2001theoretical}
\begin{equation}
H(t)=\frac{\sigma_{d}\sigma_{r}}{\sigma_{d}-\sigma_{r}}\left[\exp\left(-\frac{t}{\sigma_{d}}\right)-\exp\left(-\frac{t}{\sigma_{r}}\right)\right]\Theta(t),\label{eq:alpha}
\end{equation}
where $\sigma_{d}$ and $\sigma_{r}$ are the decay and
rise time scale, respectively, and $\Theta(\cdot)$ is the Heaviside
function. $\mathbf{A}=(A_{ij})$ is the adjacency matrix with $A_{ij}=1$
indicating a directed connection from neuron $j$ to neuron $i$ and
$A_{ij}=0$ indicating no connection from neuron $j$ to neuron $i$, $S$ is the coupling strength,
and $\tau_{jl}$ is the $l$th spike time of the $j$th neuron. 

We take the parameters as in Ref. \cite{dayan2001theoretical} that $C=1\mu\textrm{F\ensuremath{\cdot}cm}^{-2}$,
$V_{Na}=50$ mV, $V_{K}=-77$ mV, $V_{L}=-54.387$ mV, $G_{Na}=120\textrm{ mS\ensuremath{\cdot}cm}^{-2}$,
$G_{K}=36\textrm{ mS\ensuremath{\cdot}cm}^{-2}$, $G_{L}=0.3\textrm{ mS\ensuremath{\cdot}cm}^{-2}$,
and $V_{E}=0$ mV. We set synaptic time constants as $\sigma_{r}=0.5$ ms
and $\sigma_{d}=3.0$ ms. For simplicity,
we set the Poisson input parameters as $f=0.1 \textrm{ mS\ensuremath{\cdot}cm}^{-2}$
and $\nu=100$ Hz, unless indicated otherwise. However, the conclusions
shown in this work hold for a wide range of parameters corresponding to different dynamical regimes. 

When the voltage $V_{i}$ reaches the firing threshold, $V^{\textrm{th}}=-50$
mV, we say the $i$th neuron generates a spike at this time. Instantaneously,
all of its postsynaptic neurons receive this spike and the affected
change of conductance follows Eq. $\ref{eq:alpha}$. 

\subsection{HH neural network model of both excitatory and inhibitory populations}
For the HH network of both excitatory and inhibitory neurons, the dynamics of the $i$th HH neuron is also governed by Eqs. \ref{eq:HH V} and \ref{eq:HH mhn}.  But the input current $I_{i}^{\text{input}}$
is given by 
\[I_{i}^{\textrm{input}}=-G_{i}^{E}(t)(V_{i}-V_E)-G_{i}^{I}(t)(V_{i}-V_I),\]
where $G_{i}^{E}$ and $G_{i}^{I}$ are excitatory and inhibitory
conductances, respectively, $V_E$ and $V_I$ are the
corresponding reversal potentials. The conductances are defined
as 
\begin{equation*}
\begin{aligned}G_{i}^{E}(t) & =f\sum_{l}H(t-s_{il};\sigma_{d}^{E},\sigma_{r}^{E})+\sum_{j}A_{ij}S^{E}\sum_{l}H(t-\tau_{jl};\sigma_{d}^{E},\sigma_{r}^{E}),\\
G_{i}^{I}(t) & =\sum_{j}A_{ij}S^{I}\sum_{l}H(t-\tau_{jl};\sigma_{d}^{I},\sigma_{r}^{I}),
\end{aligned}
\end{equation*}
where $H$ is given in Eq. \ref{eq:alpha} and $S^E$ and $S^I$ are the excitatory and inhibitory coupling strengths, respectively.   The parameters
are set as $V_E=0$ mV, $V_I=-80$ mV, $\sigma_{r}^{E}=0.5$
ms, $\sigma_{d}^{E}=3.0$ ms, $\sigma_{r}^{I}=0.5$ ms, $\sigma_{d}^{I}=7.0$
ms. The HH neural network and the previous one of only excitatory population are efficiently simulated by an adaptive  method introduced in Ref. \cite{tian2020exponential}.

\newpage
\section*{Properties of pulse-like signals}
\begin{figure}[H]
	\begin{centering}
		\includegraphics[width=0.8\textwidth]{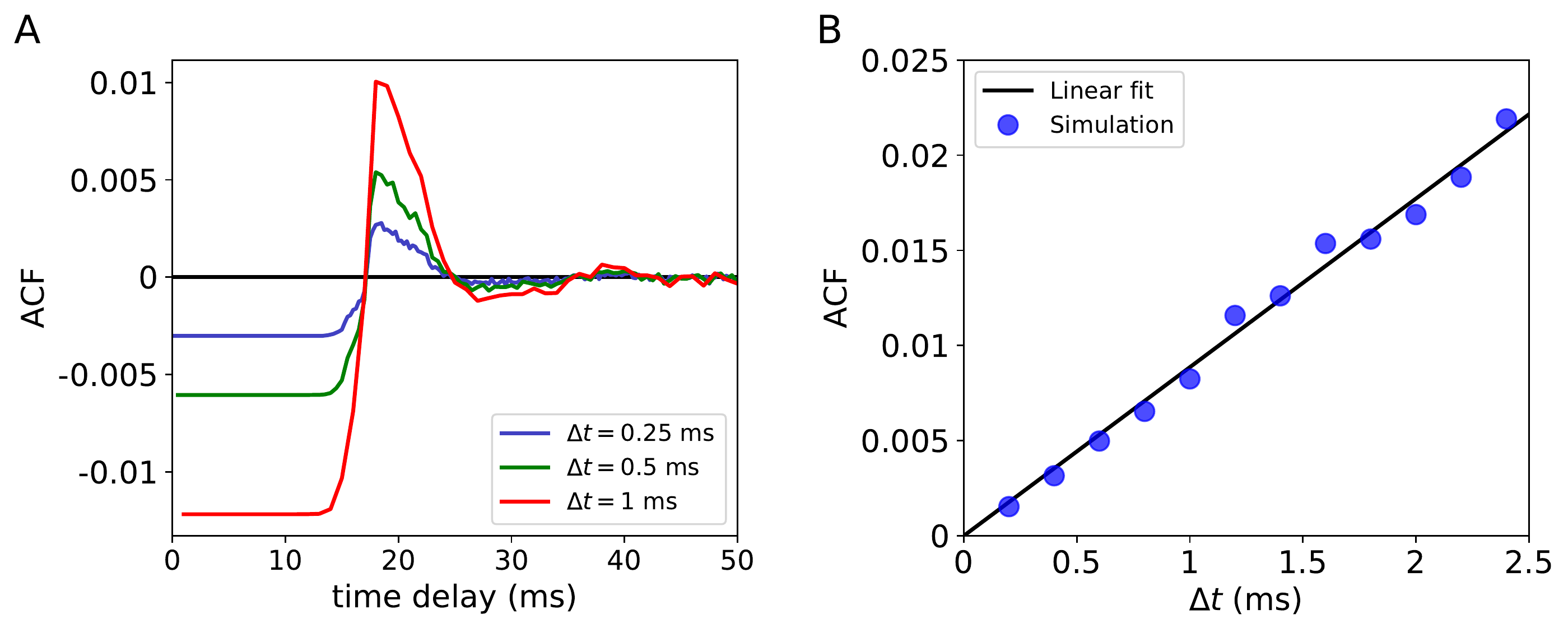}
		\par\end{centering}
	\centering{}\caption{
		Relation between ACF and sampling time bin $\Delta t$. 
		(\textit{A}) ACF curves as a function of time delay with $\Delta t=0.25$ (blue), $0.5$ (green), and 1 ms (red), respectively.
		(\textit{B}), ACF values at a fixed time delay 20 ms plotted as a function of $\Delta t$.
		The black line is a linear fit with $R^2=0.985$ which is consistent the derivation in Eq. \ref{eq:ACF vs dt}. When $\Delta{t}$ is sufficiently small, the magnitude of auto-correlation of binary time series is also small, indicating that the binary time series become almost whitened. 
		This is the results of a sample neuron driven with Poisson input $f=0.1 \textrm{ mS\ensuremath{\cdot}cm}^{-2}$
		and $\nu=100$ Hz.  
		\label{fig:SI ACF}}
\end{figure}

\begin{figure}[H]
	\begin{centering}
		\includegraphics[width=0.8\textwidth]{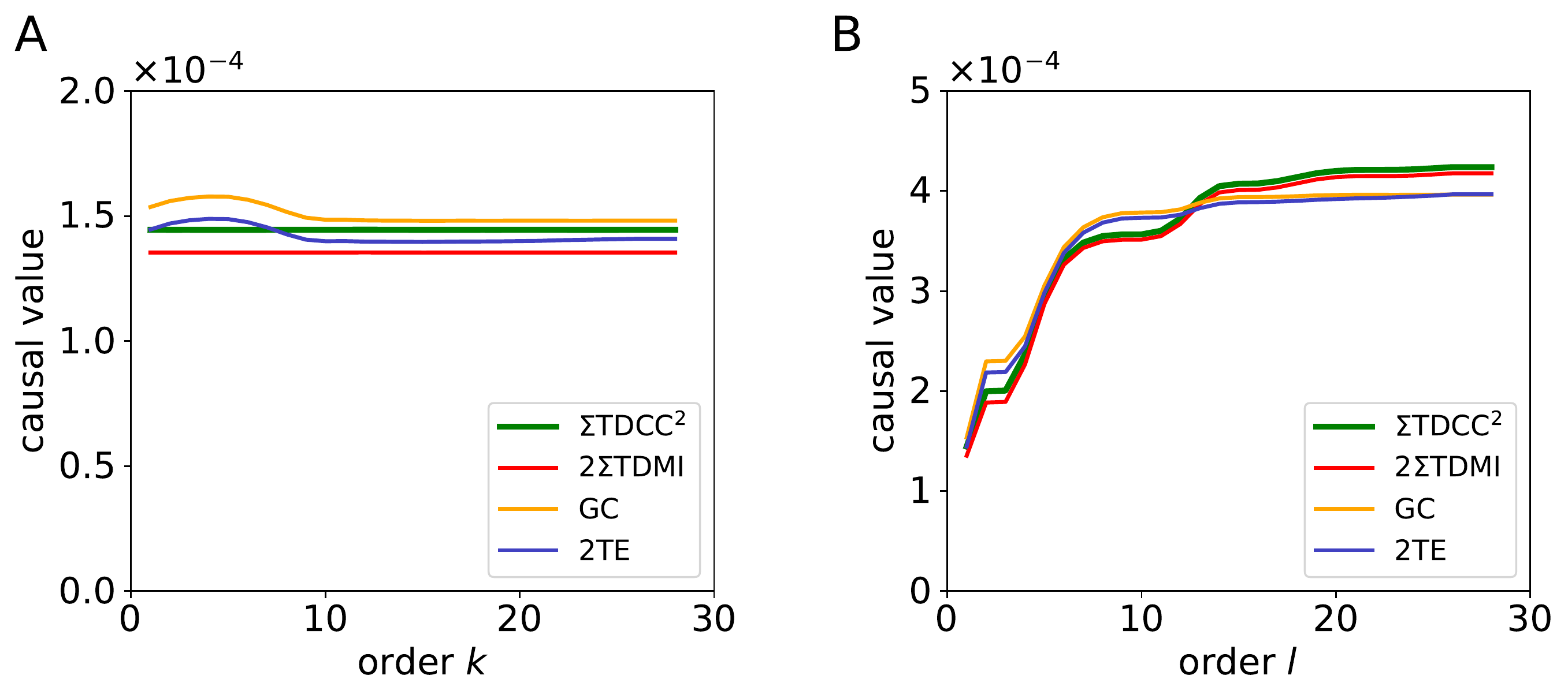}
		\par\end{centering}
	\caption{Causal values as a function of (\textit{A}) order $k$ and (\textit{B})
		order $l$ with a large $\Delta t$ obtained from causality measures from a  pair of neurons denoted by $X$ and $Y$ with unidirectional connection from $Y$ to $X$ in an HH network of 10 excitatory neurons. The HH network is randomly connected with probability 0.25. The causal values are computed from neuron $Y$ to
		neuron $X$.
		The parameters are set as $m=2$ (3 ms), $f=0.1$ $\textrm{mS\ensuremath{\cdot}cm}^{-2}$,
		$\nu=300$ Hz, $S=0.02$ $\textrm{mS\ensuremath{\cdot}cm}^{-2}$,
		$\Delta t=1.5$ ms, and $l=1$ in (\textit{A}) and $k=1$ in (\textit{B}) (all significantly greater than those of randomly surrogate time series, $p<0.05$). 
		The relation among the four causality measures revealed by theorems \textcolor{blue}{1}-\textcolor{blue}{4} in the main text  still holds when choosing the orders of $k=28$ in (\textit{A}) or $l=28$ in (\textit{B}), in both cases the event $\Bigl\Vert x_{n+1}^{(k+1)}\Bigr\Vert_{0}\geq2$ or $\Bigl\Vert y_{n+1-m}^{(l)}\Bigr\Vert_{0}\geq2$ happens with a frequency more than 44\%. This result indicates that the assumption of $\Bigl\Vert x_{n+1}^{(k+1)}\Bigr\Vert_{0}\leq1$ and $\Bigl\Vert y_{n+1-m}^{(l)}\Bigr\Vert_{0}\leq1$ is a sufficient but not necessary condition in the derivation of the quantitative relation
		between TE and TDMI.
		\label{fig:SI order kl} }
\end{figure}


\begin{figure}[H]
	\begin{centering}
		\includegraphics[width=0.5\textwidth]{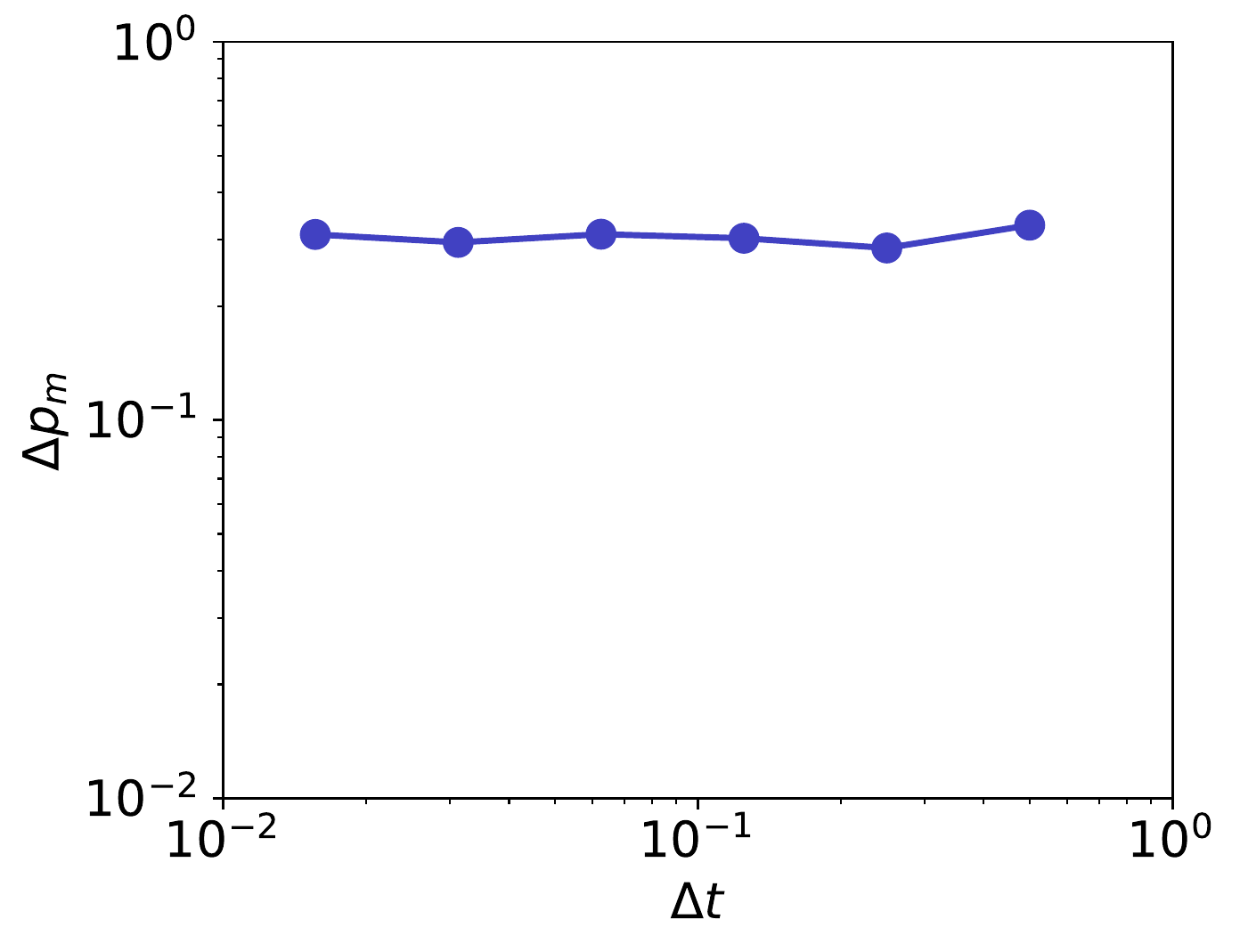}
		\par\end{centering}
	\caption{$\Delta p_{m}$ is insensitive to
		sampling resolution $\Delta t$ obtained from neuron $Y$ to
		neuron $X$ in the HH network in Fig. \ref{fig:SI order kl}. The parameters are set as $f=0.1 \textrm{ mS\ensuremath{\cdot}cm}^{-2}$, $\nu=100$ Hz, a fixed time delay 3 ms and
		$S=0.02$ $\textrm{mS\ensuremath{\cdot}cm}^{-2}$.}
\end{figure}

\newpage
\section*{Consistency among causality measures across different dynamical regimes}

\begin{figure}[H]
	\begin{centering}
		\includegraphics[width=0.55\textwidth]{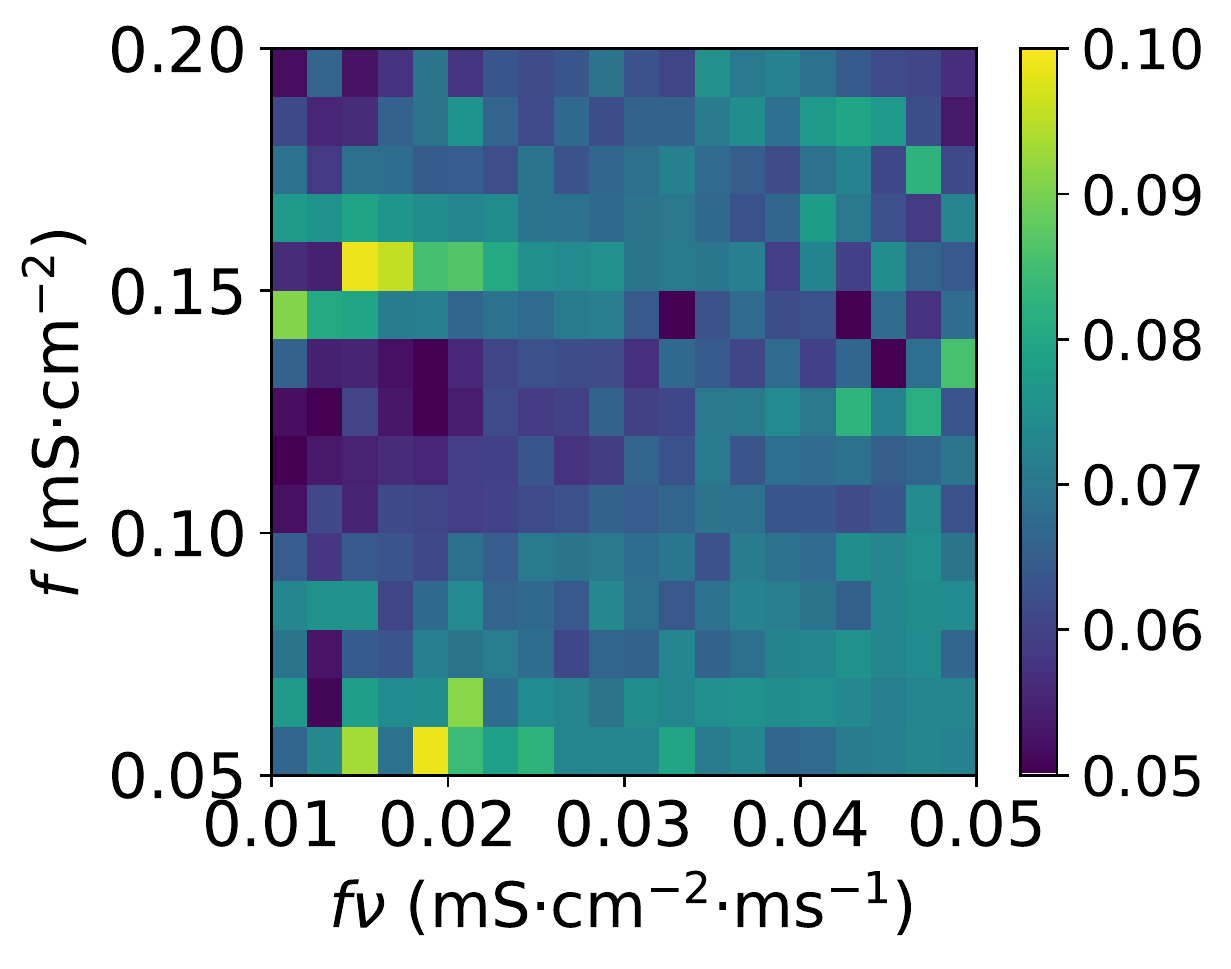}
		\par\end{centering}
	\caption{Relative error of causal values with different external Poisson input parameters $f$ and $\nu$ obtained from neuron $Y$ to
		neuron $X$ in the HH network in Fig. \ref{fig:SI order kl}. 
		The relative error is computed by 
		$\frac{\max\{\Sigma\text{TDCC}^{2},2\Sigma\text{TDMI},\text{\text{GC},\ensuremath{\text{TE}}}\}-\min\{\Sigma\text{TDCC}^{2},2\Sigma\text{TDMI},\text{\text{GC},\ensuremath{\text{TE}}}\}}{\max\{\Sigma\text{TDCC}^{2},2\Sigma\text{TDMI},\text{\text{GC},\ensuremath{\text{TE}}}\}}$
		and indicates that the mathematical relations revealed by theorems \textcolor{blue}{1}-\textcolor{blue}{4} in the main text hold for a wide range of Poisson input parameters. 
		Other parameters are set as $\Delta t=0.5$ ms, $k=l=1$, $S=0.01$ $\textrm{mS\ensuremath{\cdot}cm}^{-2}$, 
		and $m=6$ (3 ms).
	}
\end{figure}

\newpage
\section*{Reconstruction of structure connectivity for asynchronous state}
\begin{figure}[H]
	\begin{centering}
		\includegraphics[width=1\textwidth]{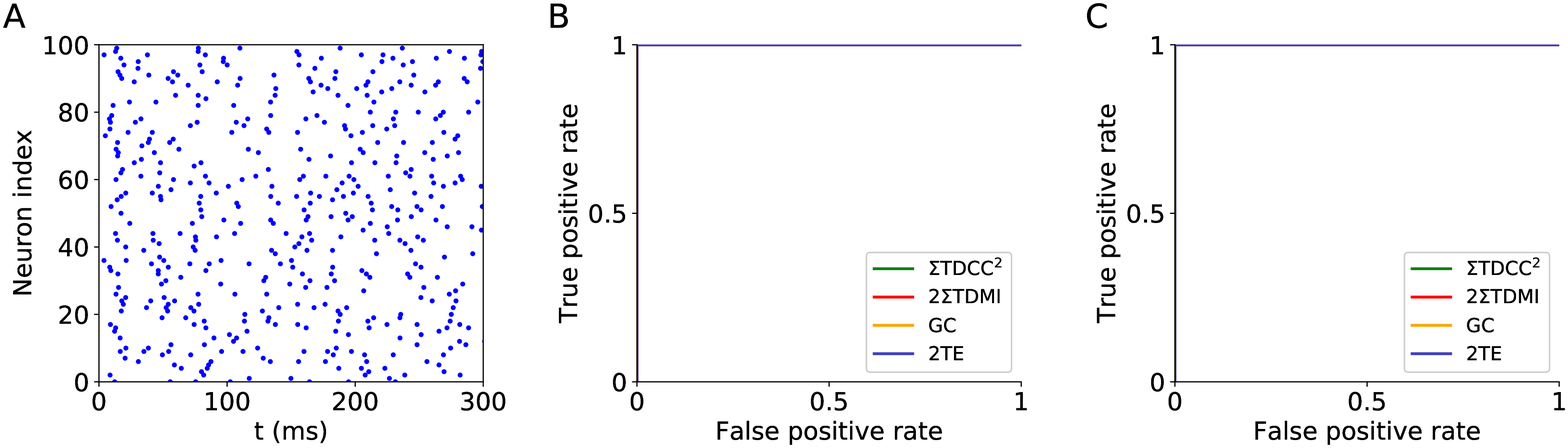}
		\par\end{centering}
	\caption{Performance of the causality measures in an HH network in the asynchronous state. The network is composed of 100 excitatory neurons randomly connected
		with probability 0.25. (\textit{A}) Raster plot of neuronal firing indicating that the network is in an asynchronous state.
		(\textit{B}) ROC curves of the full
		HH network with $\text{AUC}=1$. (\textit{C}) ROC curves of an HH subnetwork of 20 neurons with $\text{AUC}=1$.
		The green curve is for $C^{2}(X,Y;m)$, the red
		curve is for $2I(X,Y;m)$, the orange curve is for $G_{Y\rightarrow X}(k,l;m)$,
		and the  blue curve is for $2T_{Y\rightarrow X}(k,l;m)$.
		The ROC curves for TDCC, TDMI, GC, and TE overlap with each other.
		The parameters are set as $f=0.1 \textrm{ mS\ensuremath{\cdot}cm}^{-2}$, $\nu=100$ Hz, $\Delta t=0.5$ ms, $k=l=1$, $S=0.02$
		$\textrm{mS\ensuremath{\cdot}cm}^{-2}$, and $m=6$ (3
		ms).\label{fig:SI ROC}}
\end{figure}

\newpage
\section*{Dependence of $\delta p$ on $S$}
\begin{figure}[H]
	\begin{centering}
		\includegraphics[width=0.5\textwidth]{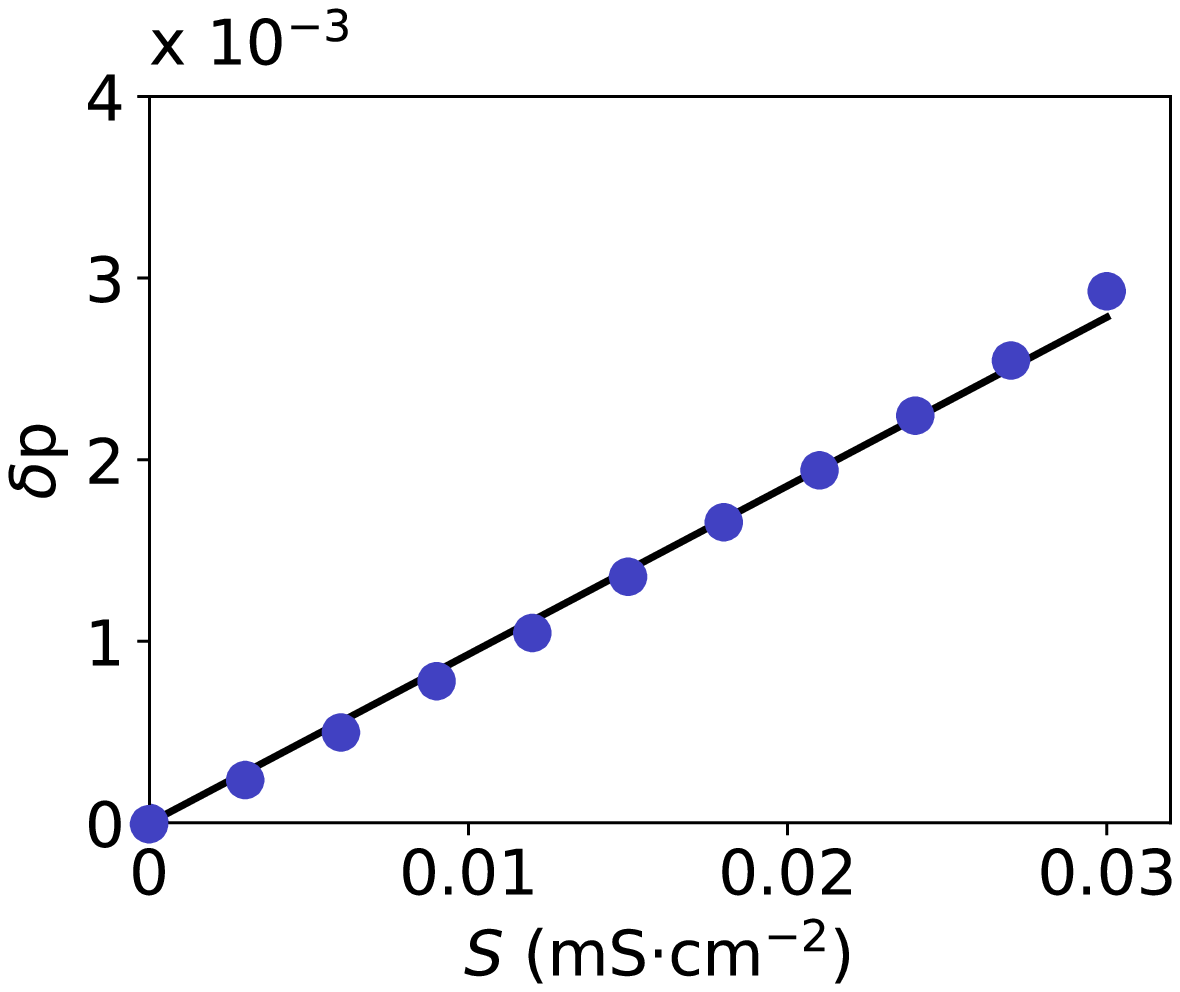}
		\par\end{centering}
	\caption{$\delta{p_{Y\rightarrow X}}=p(x_{n}=1|y_{n-m}=1)-p(x_{n}=1|y_{n-m}=0)$ as a function of the coupling strength $S$ from neuron $Y$ to neuron $X$ in an HH network of two excitatory neurons with unidirectional connection from $Y$ to $X$, which measures the increment of probability  of generating a pulse output by neuron $X$ at time step $n$ induced by a pulse-output  signal of neuron $Y$ at an earlier time step $n-m$. The black line is a linear fit with $R^2=0.996$. The parameters are set as $f=0.1 \textrm{ mS\ensuremath{\cdot}cm}^{-2}$, $\nu=100$ Hz, $\Delta t=0.5$ ms, and $m=6$ (3 ms).}
\end{figure}

\newpage
\section*{Reconstruction of structure connectivity with experimental data}
\begin{figure}[H]
	\centering{}\includegraphics[width=1\textwidth]{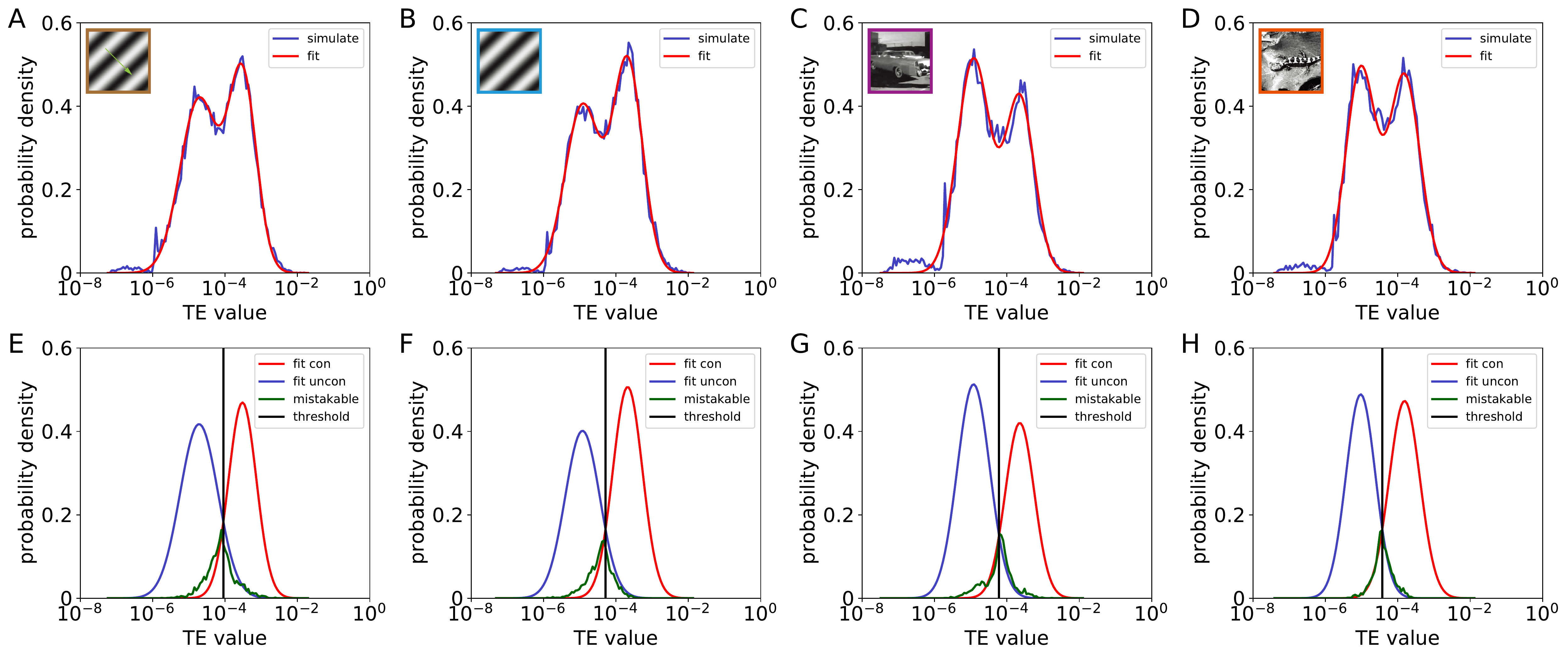}
	\caption{Reconstruction of structural connectivity by the assumption of log-normal
		distribution of causal values for experimental spike data. (Top panel): Distribution of TE values in the network composed by the observed neurons in experiment under visual stimuli of (\textit{A}) drifting gratings, (\textit{B}) natural movie, (\textit{C}) static gratings, and (\textit{D}) natural scenes. The blue and red curves are the computed and fitted distributions, respectively. (Bottom panel): Distributions of fitted TE values from connected (red) and unconnected (blue) pairs which are obtained from the fitting in top panel. The black vertical line is the optimal inference threshold and the green curve is the mistakable causal values. We use the experimental spike data (sections id 715093703 at \url{https://allensdk.readthedocs.io/}) with signal-to-noise ratio greater than 4 and firing rate greater than 0.05 Hz. The parameters are set as  $k=1,l=10$, $\Delta t=1$ ms, and $m=1$. \label{fig:allen}
	}
\end{figure}

\newpage
\section*{Verify the log-normal distributed assumption for causality measures}
\begin{figure}[H]
	\centering{}\includegraphics[width=0.8\textwidth]{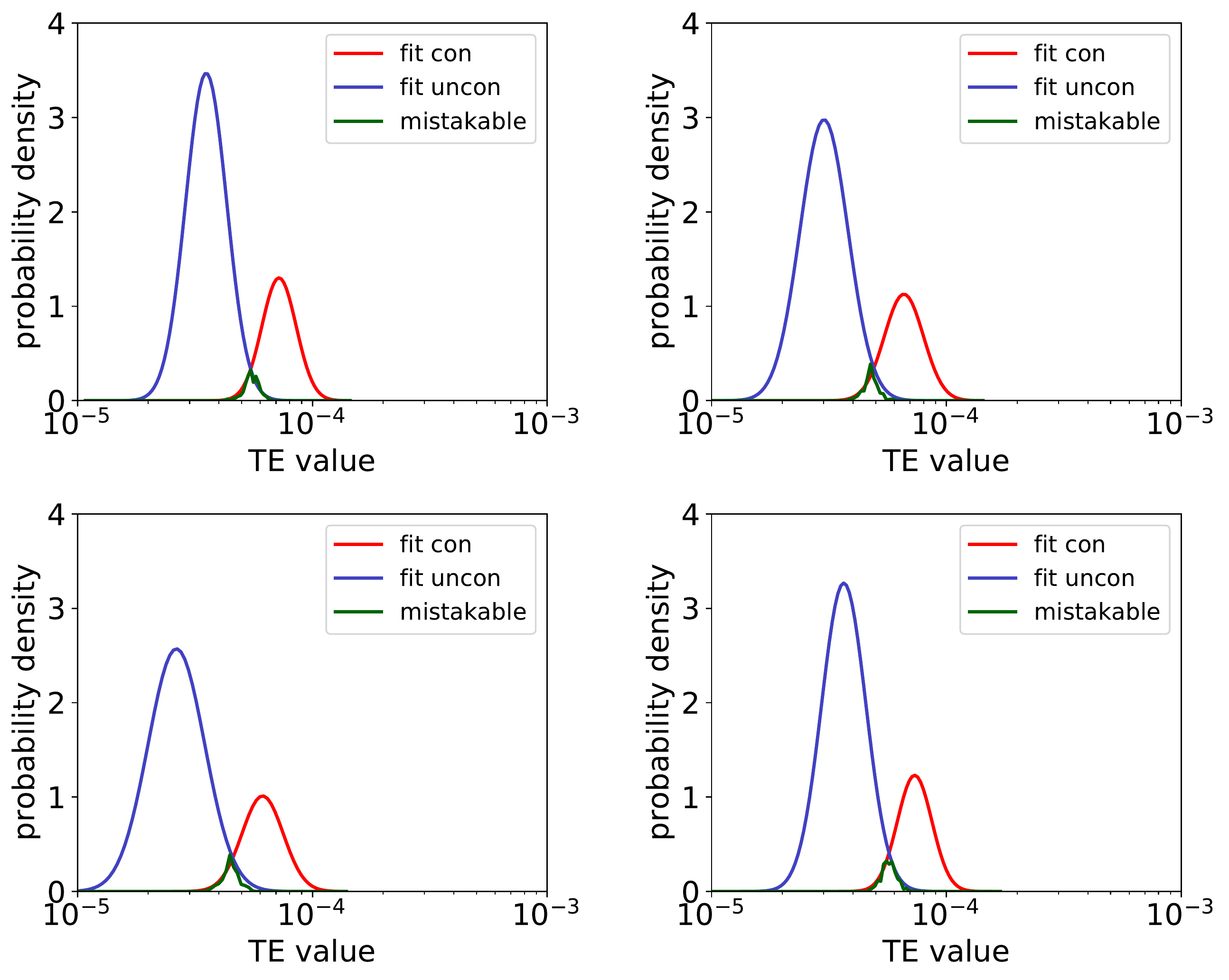}
	\caption{Reconstruction of structural connectivity by the assumption of log-normal
		distribution of causal values for an HH network of 100
		excitatory neurons receiving correlated external Poisson input. The correlation coefficient of the Poisson input to each neuron is generated to be 0.3. 
		The parameters are the same as those in Fig. \ref{fig:SI ROC} except the Poisson input rate which is (\textit{A}) $\nu=90$ Hz,
		(\textit{B}) $\nu=100$ Hz, (\textit{C}) $\nu=110$ Hz, and (\textit{D}) $\nu=120$ Hz. The colors are the same as those in Fig. \ref{fig:allen}.
	}
\end{figure}

\newpage
\section*{Continuous-valued signals breaks the mathematical relations among four causality measures}
\begin{figure}[H]
	\begin{centering}
		\includegraphics[width=0.8\textwidth]{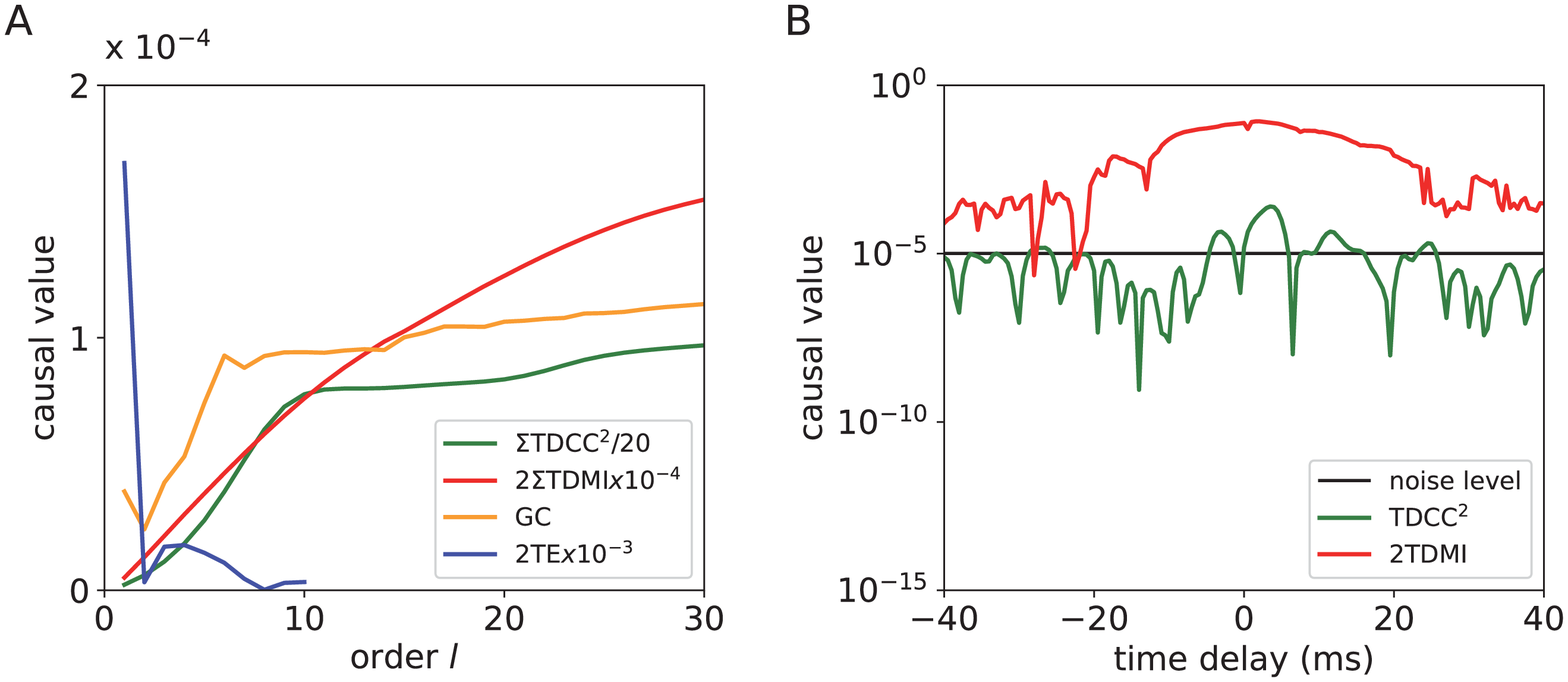}
		\par\end{centering}
	\caption{The break of the mathematical relations among the causality measures for continuous-valued voltage time series obtained from neuron $Y$ to neuron $X$ in the HH network in Fig. \ref{fig:SI order kl}.
		(\textit{A}) TDCC, TDMI, GC, and TE  as a function
		of order $l$ computed from continuous-valued voltage time series. The order $l$
		for TE is cut off at 10 due to the exponential increase of computational cost. (\textit{B})
		TDCC and TDMI as a function of time delay with positive (negative)
		delay corresponding to the calculation of causal values from $Y$ to $X$ (from $X$ to
		$Y$). The black line represents the noise level, which is obtained as the
		largest value of TDCC (TDMI) after shuffling the time series and computing TDCC (TDMI) between the shuffled signals for 100 times. A bidirectional connection between $X$ and $Y$ will be incorrectly inferred by TDMI due to the strong self-correlation of the continuous-valued voltage time series. 
		The parameters are set as $f=0.1 \textrm{ mS\ensuremath{\cdot}cm}^{-2}$, $\nu=100$ Hz, order $k=l$ and $m=1$ in (\textit{A}), and  $S=0.02$ $\textrm{mS\ensuremath{\cdot}cm}^{-2}$, $\Delta t=0.5$ ms in (\textit{A}) and (\textit{B}).
		\label{fig:SI conti TGIC}}
\end{figure}

\newpage
\section*{Reconstruction of structure connectivity in more general situations}

\begin{figure}[H]
	\begin{centering}
		\includegraphics[width=1\textwidth]{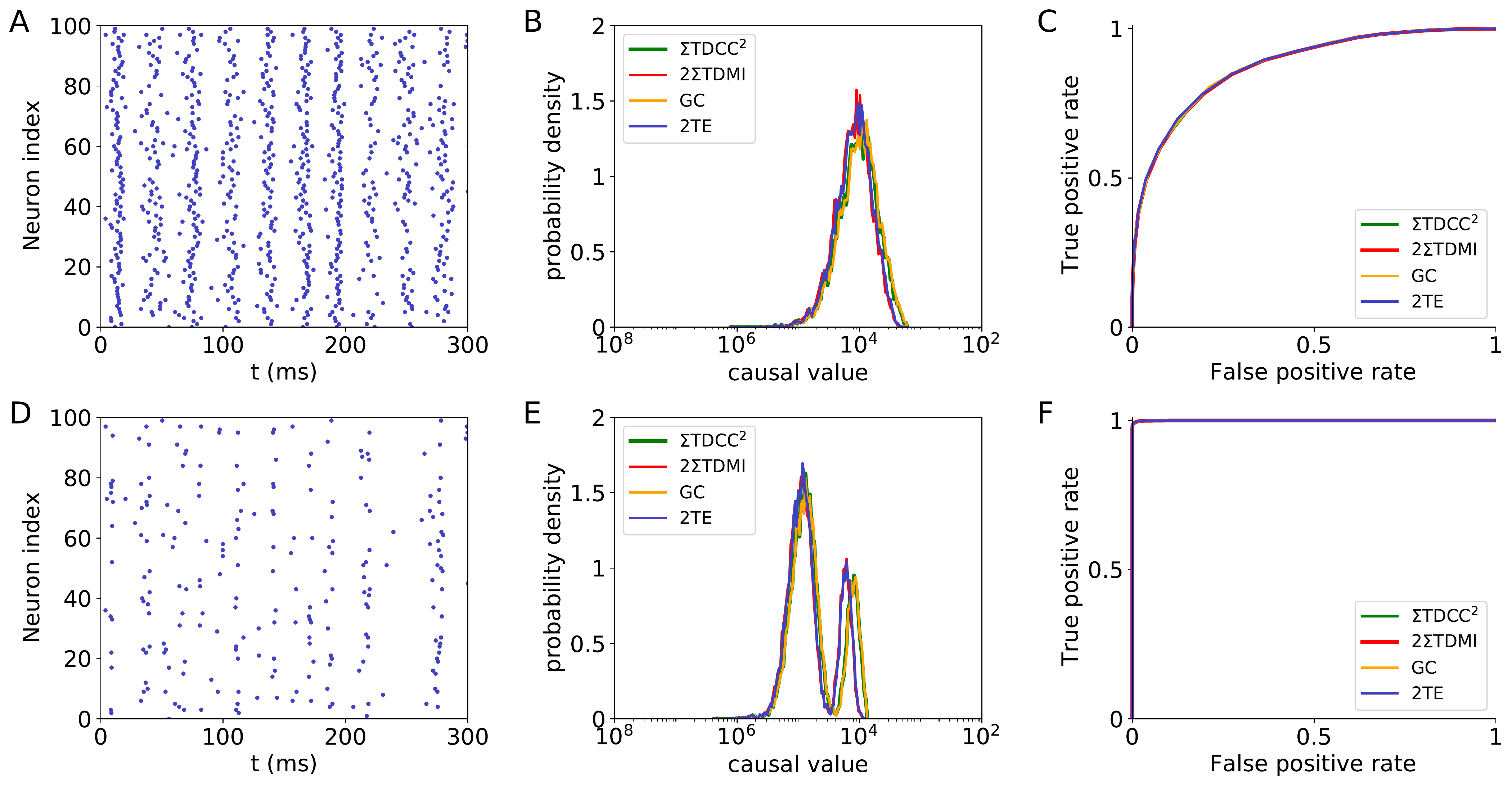}
		\par\end{centering}
	\caption{Performance of the causality measures in an HH network of 100 excitatory neurons in the nearly
		synchronous state. (Top panel): Results using the original spike train.
		(Bottom panel): Results using the spike train from desynchronized sampling that only samples the pulse-output 
		signals in the asynchronous time interval. (\textit{A,
			D}): Raster plot of the neuronal firing. (\textit{B, E}): The distribution of causal values
		of each pair of neurons in the whole network.
		(\textit{C, F}): ROC curves of the HH network with $\text{AUC}=0.88$ (upper) and $\text{AUC}=0.99$ (lower).
		The ROC curves for TDCC, TDMI, GC, and TE nearly overlap.
		The colors and parameters are the same as those in Fig. \ref{fig:SI ROC},
		except that the coupling strength $S=0.028$ $\textrm{mS\ensuremath{\cdot}cm}^{-2}$.
		\label{fig:SI synchronize}}
\end{figure}

\begin{figure}[H]
	\begin{centering}
		\includegraphics[width=0.5\textwidth]{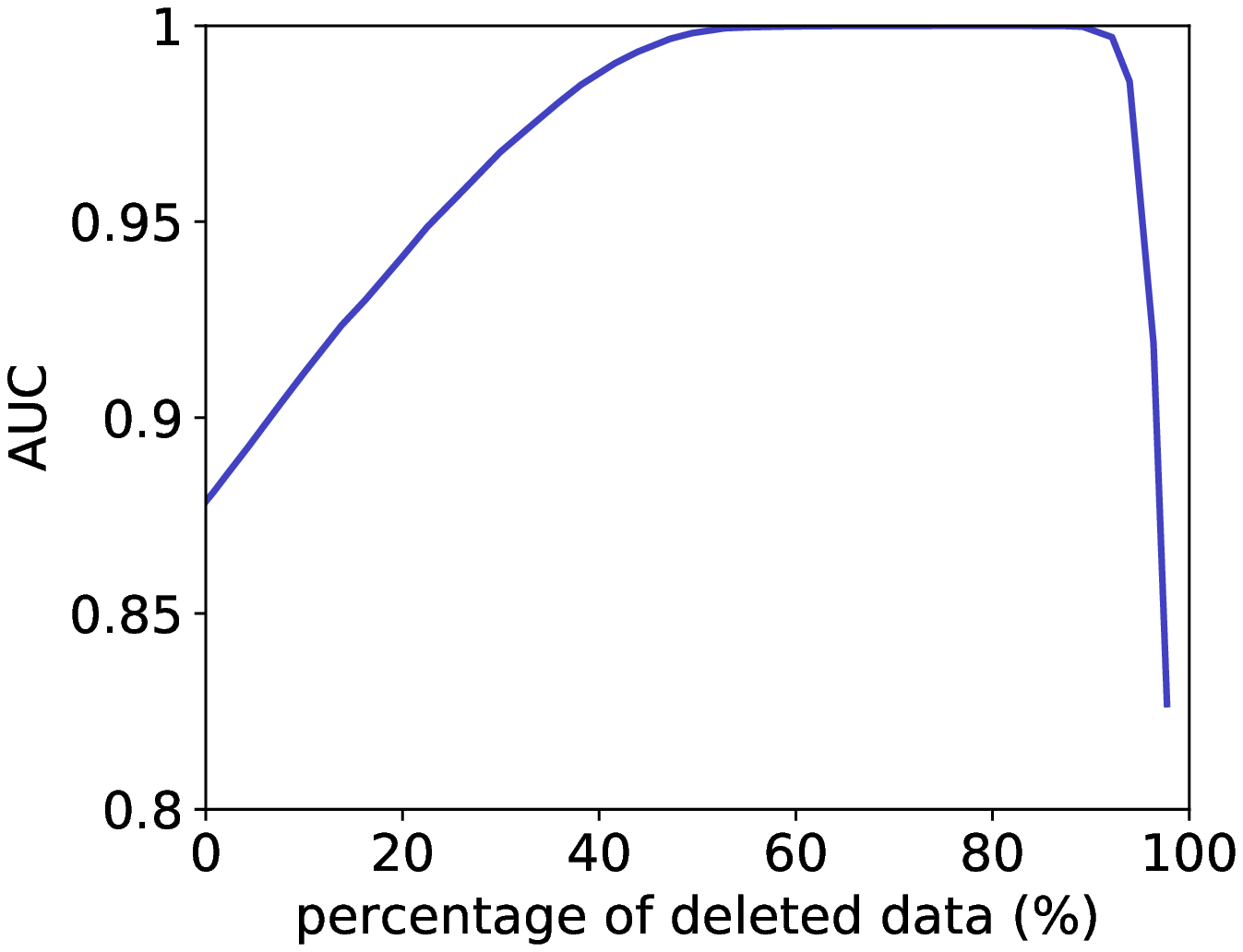}
		\par\end{centering}
	\caption{AUC as a function of percentage of deleted data in the spike train of the HH network in Fig. \ref{fig:SI synchronize}A.
		78 $\%$ of the spike data (in the more synchronous time interval) are deleted in  Fig. \ref{fig:SI synchronize}D. }
\end{figure}

\newpage
\begin{figure}[H]
	\begin{centering}
		\includegraphics[width=1\textwidth]{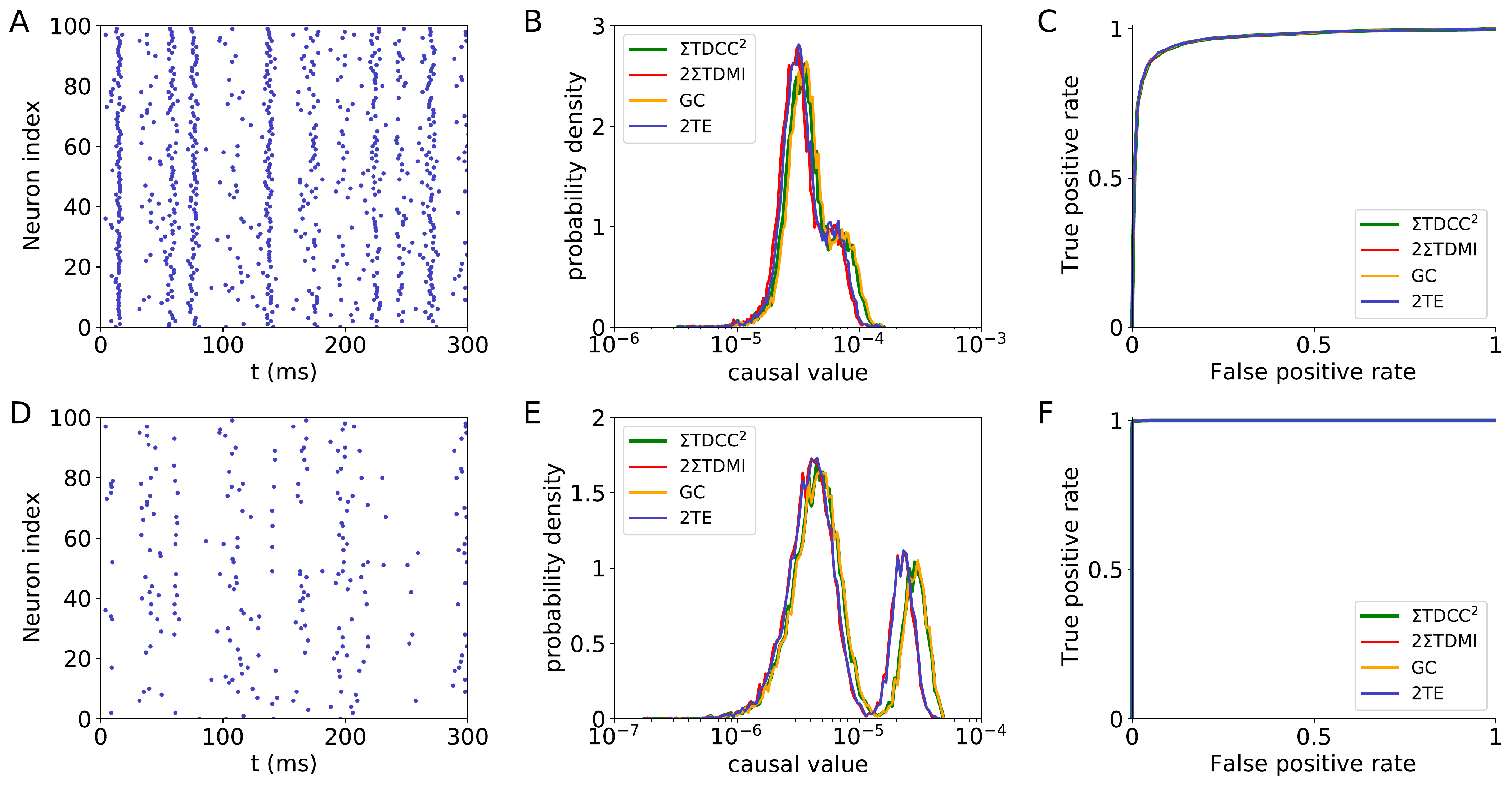}
		\par\end{centering}
	\caption{Performance of the causality measures in an HH network of 100 excitatory neurons receiving correlated
		external Poisson inputs. The correlation coefficient of the Poisson input to each neuron is generated to be 0.35.
		(Top panel): Results using the original spike
		train. (Bottom panel): Results using the spike train from desynchronized sampling.
		(\textit{A, D}): Raster plot of the neuronal firing. (\textit{B, E}):
		The distribution of causal values
		of each pair of neurons in the whole network. (\textit{C, F}): ROC curves of the
		HH network with $\text{AUC}=0.96$ (upper) and $\text{AUC}=0.99$ (lower). The ROC curves for TDCC, TDMI, GC, and TE nearly
		overlap. The colors and parameters are the same as those in Fig.
		\ref{fig:SI ROC}. \label{fig:SI common Poisson input}}
\end{figure}

\begin{figure}[H]
	\begin{centering}
		\includegraphics[width=1\textwidth]{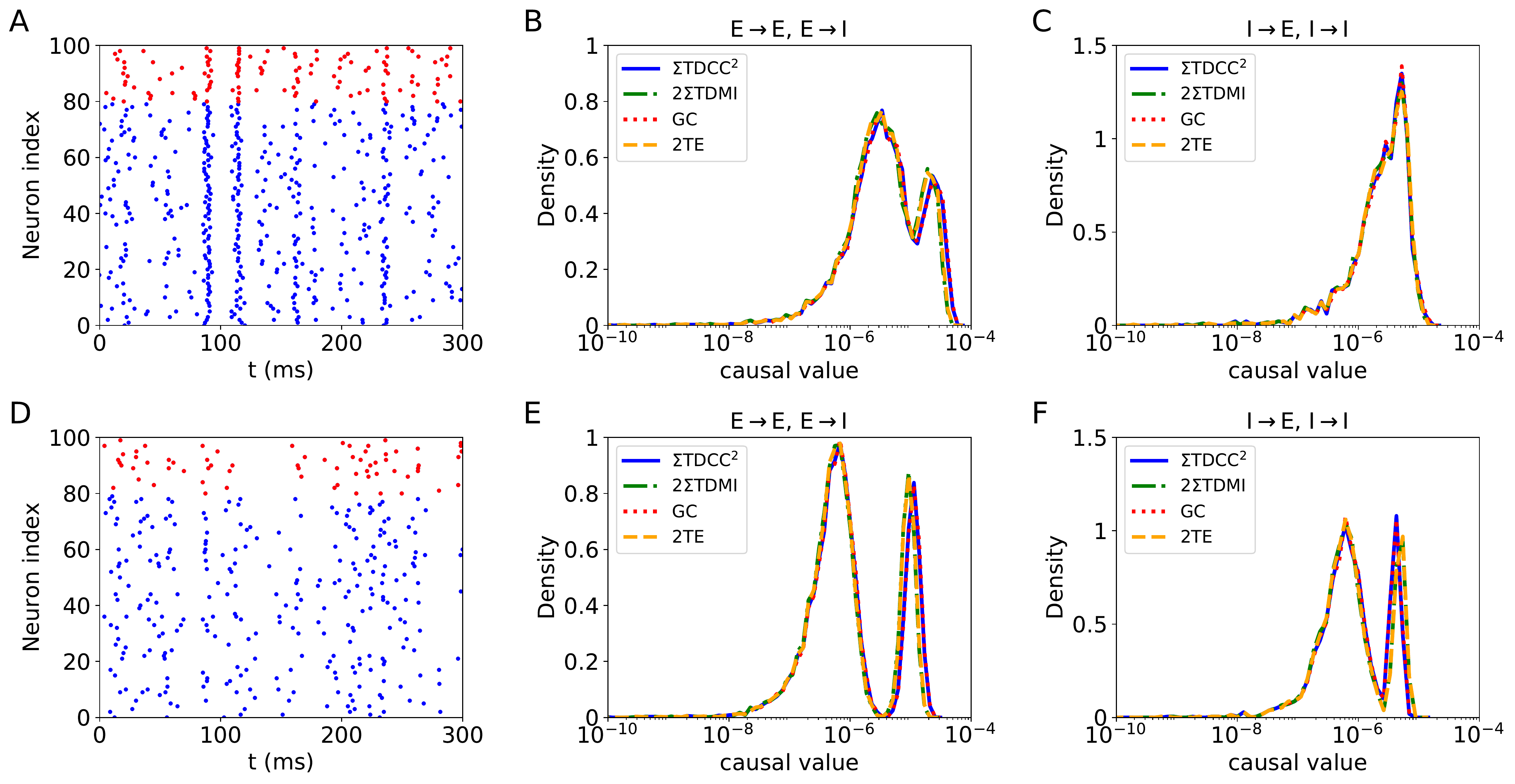}
		\par\end{centering}
	\caption{Performance of the causality measures in an HH network of 80 excitatory and 20 inhibitory neurons randomly connected with probability 0.25.  
		(Top panel): Results using the original spike train.
		(Bottom panel): Results using the spike train from desynchronized sampling.
		(\textit{A, D}) Raster plot of the neuronal firing. The blue and red dots indicate the excitatory and inhibitory neurons, respectively. (\textit{B, E})
		The distribution of causal values
		of each pair of neurons with the presynaptic neuron being excitatory. (\textit{C, F}) The distribution of causal values
		of each pair of neurons with the presynaptic neuron being inhibitory. The colors and parameters are the same as those in Fig.
		\ref{fig:SI ROC}.  The AUC values for (B, C, E, F) are 0.96, 0.71, 1, and 0.99, respectively. 
		The coupling strength is $S^E=0.02$ $\textrm{mS\ensuremath{\cdot}cm}^{-2}$ and $S^I=0.08$ $\textrm{mS\ensuremath{\cdot}cm}^{-2}$. The correlation coefficient of the Poisson input to each neuron
		is generated to be 0.15. \label{fig:SI EI network}}
\end{figure}

\begin{figure}[H]
	\centering{}\includegraphics[width=1\textwidth]{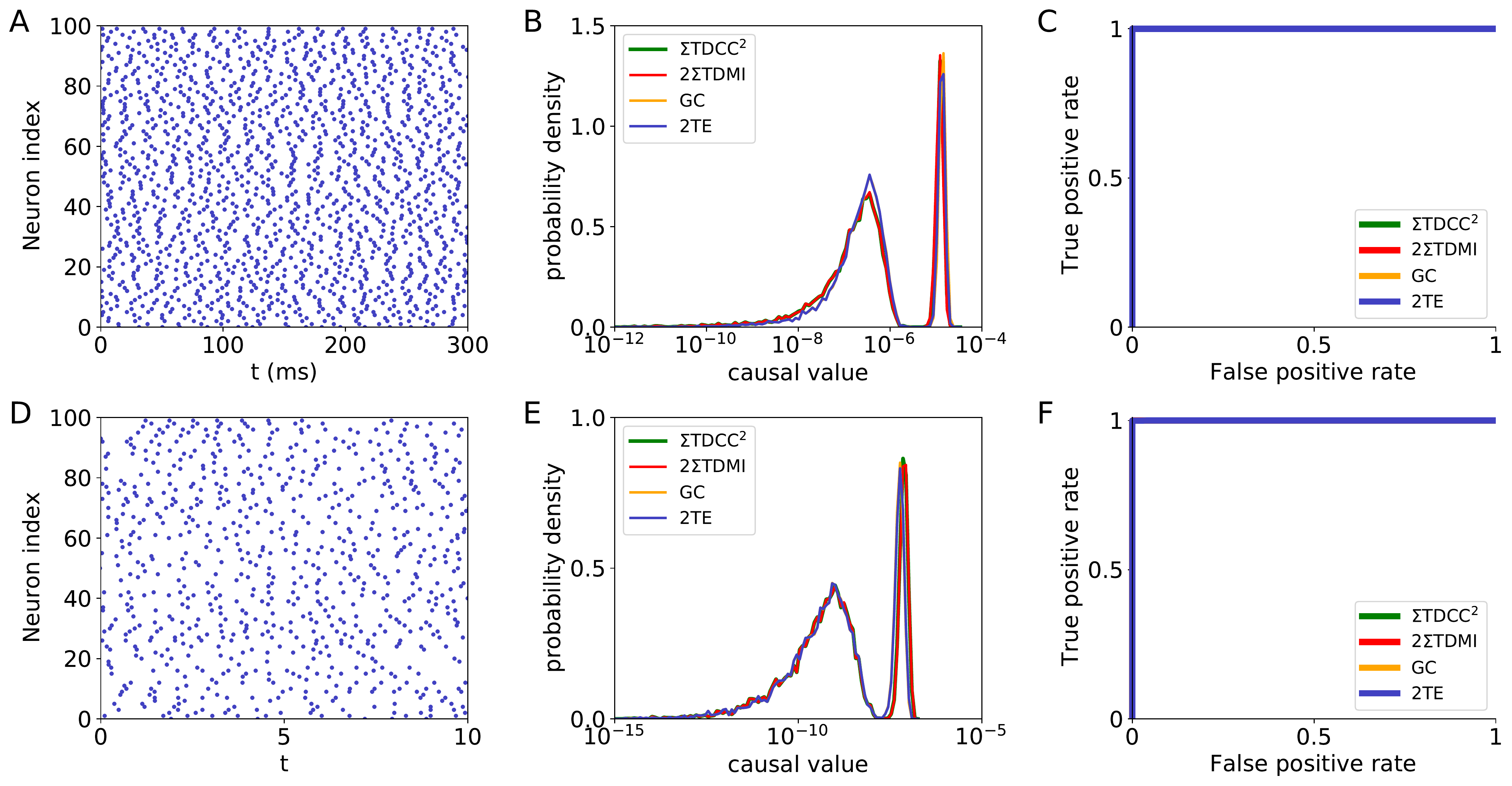}
	\caption{(Top panel): Performance of the causality measures in an I\&F network of 100 excitotary neurons randomly
		connected with probability 0.25. (\textit{A}) Raster plot of the neuronal firing. (\textit{B}) The distribution of causal values
		of each pair of neurons in the whole network. (\textit{C}) ROC curves of the LIF network in (\textit{A}) with $\text{AUC}=1$.
		The parameters are set as $f=0.05,\nu=1.5$ KHz, $S=0.01$, $\Delta t=0.5$
		ms, $m=1$, and orders $k=l=1$. (Bottom panel): Performance of the causality measures in a Lorenz network of 100 nodes randomly connected as in (\textit{A}).
		(\textit{D}) Raster plot of the pulse firing.
		(\textit{E}) The distribution of causal values
		of each pair of nodes in the whole network. (\textit{F}) ROC curves
		of the Lorenz network in (\textit{D}) with $\text{AUC}=1$. The parameters are set as $S=0.25$,
		$\Delta t=0.01$, $m=1$, and orders $k=l=1$. The ROC
		curves for TDCC, TDMI, GC, and TE in (\textit{C}) and (\textit{F})
		overlap with each other. The colors are the same as those in Fig.
		\ref{fig:SI ROC}. \label{fig:SI N=00003D100}}
\end{figure}

\begin{small}	
\bibliographystyle{unsrt}
\bibliography{reference_PNAS}	
\end{small}
